\documentclass[11pt,a4paper]{article}

\usepackage{url}
\usepackage{graphicx}
\usepackage{bm}
\usepackage[all,cmtip]{xy}
\usepackage{comment}
\usepackage{color}
\usepackage[american]{isodate}
\usepackage{amsmath}
\usepackage{amsfonts}
\usepackage{amssymb}
\usepackage{amsthm}
\usepackage{authblk}
\usepackage{fullpage}
\usepackage{tikz-cd}
\usepackage{authblk}

\newcommand\independent{\protect\mathpalette{\protect\independenT}{\perp}}
\def\independenT#1#2{\mathrel{\rlap{$#1#2$}\mkern2mu{#1#2}}}

\renewcommand{\vec}[1]{{\bm{#1}}}

\newcommand{\E}{\mathbb{E}}

\newcommand{\fig}[1]{Fig.~(\ref{fig:#1})}
\newcommand{\eq}[1]{Eq.~(\ref{eq:#1})}
\renewcommand{\a}{\alpha}
\newcommand{\g}{\gamma}
\newcommand{\e}{\varepsilon}
\renewcommand{\b}{\beta}
\renewcommand{\l}{\lambda}
\newcommand{\sg}{\sigma}

\renewcommand{\k}{\kappa}
\renewcommand{\t}{\theta}
\newcommand{\Tr}{\text{Tr}}

\newcommand{\RH}{\mathbb{R}}
\newcommand{\I}{\mathcal{I}}
\renewcommand{\O}{\mathcal{O}}

\newcommand{\set}[1]{\lbrace #1 \rbrace}

\newcommand{\ra}{\rightarrow}

\newcommand{\nn}{\nonumber}

\numberwithin{equation}{section}      
\theoremstyle{plain}
\newtheorem{thm}{Theorem}[section]  
\newtheorem{lem}[thm]{Lemma} 
\newtheorem{ass}[thm]{Assumption}
\newtheorem{prop}[thm]{Proposition} 
\newtheorem{cor}[thm]{Corollary}

\theoremstyle{definition}
\newtheorem{defn}{Definition}[section]
\newtheorem{exmp}{Example}[section]

\newtheorem{rem}{Remark}[section]

\begin{document}

\title{Volatility Inference and Return Dependencies in Stochastic Volatility Models}
\author[1]{Oliver Pfante \thanks{pfante@fias.uni-frankfurt.de}}
\author[1]{Nils Bertschinger \thanks{bertschinger@fias.uni-frankfurt.de}}

\affil[1]{Systemic Risk Group, Frankfurt Institute for Advanced Studies}

\date{\today}
\maketitle

\begin{abstract}
 Stochastic volatility models describe stock returns $r_t$ as driven
  by an unobserved process capturing the random dynamics of volatility
  $v_t$. The present paper quantifies how much information about volatility
  $v_t$ and future stock returns can be inferred from past returns in stochastic 
  volatility models in terms of Shannon's mutual information.
\end{abstract}

\section{Introduction}

Many plaudits have been aptly used to describe Black and Scholes'
\cite{Black1973} contribution to option pricing theory. However,
especially after the 1987 crash, the geometric Brownian motion model
and the Black-Scholes formula were unable to reproduce the option
price data of real markets. This is not surprising, since the Black-Scholes 
model makes the strong assumption that log-returns of stocks are independently 
normally distributed with a volatility which is not only assumed to be known but 
also constant over time. Both assumptions are wrong: first, there are long range 
dependencies among returns; second, volatility is a hidden parameter which
needs to be inferred from stock and option data, respectively, and is not constant 
at all but a highly volatile time-process.  Among the most relevant
statistical properties of these volatility stochastic
processes, volatility seems to be responsible for the observed
clustering in stock returns. That is, large returns are commonly
followed by other large returns and similarly for small ones
\cite{Bouchaud2003}. Another feature is that, in clear contrast with
stock returns, which show negligible autocorrelations, squared stock return,
which is essentially volatility, autocorrelation is still significant for time 
lags longer than one year \cite{Perello2008,Muzy2000,Bouchaud2003,
Lo1991,Ding1993,LeBaron2001}. Additionally, there exists the so-called 
leverage effect, i.e., much shorter (few weeks) negative cross-correlation 
between current stock returns and future volatility \cite{Bouchaud2001,Bouchaud2003,
Black1973,Bollerslev2006}.

Inspired by these observations, stochastic volatility models have been developed which describes 
time-varying volatility as well as dependencies among stock returns. Here, we build 
on a range of generic stochastic volatility models empirically studied in \cite{sixFactor} and the Exponential
Ornstein-Uhlenbeck model. Stochastic volatility models describe volatility as a stochastic 
process in its own right. This process then couples to the stock price process, thereby
trying to capture many prominent statistical properties of stock returns, temporal clustering, 
leveraging, volatility autocorrelation, as well as the volatility smile observed in option 
prices. In the present paper we pose two question on stochastic volatility models: first, 
how reliably can the hidden volatility be inferred from observable stock-return data; second, 
how dependent are subsequent returns in these models? Since, in stochastic volatility models stock prices as well as the time
varying volatility are both modelled as random variables, Shannon's
information theory \cite{Shannon1948} provides an ideal frame to make
this question precise. In particular, in this paper we address the
question of how much information observed stock returns provide about
the hidden volatility and future returns.

To this end, we first introduce basic information theory in section
\ref{sec:InfoTheo}. In particular, we recall the definition of mutual
information, a general measure of statistical dependence and provide
an intuitive interpretation in the context of volatility estimation. In
the subsequent subsection we recall the basics on existence
and uniqueness of stationary solutions of the Fokker-Planck 
equation. Furthermore, we introduce the logarithmic Sobolev
inequality which is vastly applied in the present paper. This inequality has
played a major role in recent mathematics in estimating convergence
rates of arbitrary time-dependent solutions of the Fokker-Planck
equation towards the stationary one (if it exists) \cite{Villani}.
In section \ref{sec:StochVolInfo} we explain how this allows to
quantify the amount of information obtained in stochastic volatility
models and derive analytic upper bounds for a wide class of single
factor stochastic volatility models in section
\ref{sec:StochVolModels}. Using realistic parameter values fitted from
the literature, we find that stock returns provide in general only limited
information about the volatility and are fairly independent in 
stochastic volatility models on any time-scale.  Our calculations 
reveal that this is not a data issue, but instead arises for quite fundamental
information theoretic reasons. This not only implies that
volatility estimates from stock data are inherently imprecise, 
but, as we believe, also has severe implications for volatility 
predictions and sheds doubt on the standard practice of comparing 
models based on their forecasting error \cite{Hansen2005}. The 
inability of stochastic volatility models to impose strong return
dependencies suggests that their less complex single factor Jump-Diffusion
models \cite{Gatheral2010} and L\'{e}vy models \cite{Schoutens2003}, respectively, 
provide a reasonable alternative. Even though returns are independent 
in these models, they allow for large tails and lead to analytically tractable 
option pricing models that can generate implied volatility smiles.
We conclude with a short comment on multi-factor stochastic volatility 
models.

\section{Prerequisites}

This section introduces basic elements of information theory and results on the Fokker Planck-equation.

\subsection{Information Theory}
\label{sec:InfoTheo}

The \textit{differential entropy} of a tuple $\vec X = (X_1, X_2, \ldots, X_n)$ of random variables with joint density $f(\vec x) = f(x_1, x_2,$ $ \ldots, x_n)$ is defined as \cite{Cover2006}
\[
h(\vec X) = - \int f \log f \, 
\]
which might be $- \infty$. Here, and in the sequel, we generally drop the $d \vec x$-notation in integrals w.r.t. the Lebesgue measure as long as no confusions occur and shall write $dp$ if we integrate w.r.t. a measure $p$. 

\begin{exmp} \label{diffEntGauss}
Let $\vec X$ be a $n$ dimensional and normally distributed with mean $\vec \mu$ and covariance matrix $A$. We compute
\begin{align*}
h(\vec X) &=   - \int \dfrac{1}{\sqrt{2^n\pi^n \det(A)}}e^{- \frac{1}{2} (\vec x - \vec \mu)^T A^{-1} (\vec x - \vec \mu) } \log \left(\dfrac{1}{\sqrt{2^n\pi^n \det(A)}}e^{- \frac{1}{2} (\vec x - \vec \mu)^T A^{-1}  (\vec x - \vec \mu) } \right) \\
&= \dfrac{1}{2} \log \left( 2^n\pi^n \det(A)\right) + \int \dfrac{1}{2} (\vec x - \vec \mu)^T A^{-1}  (\vec x - \vec \mu) \dfrac{1}{\sqrt{2^n\pi^n \det(A)}}e^{- \frac{1}{2} (\vec x - \vec \mu)^T A^{-1}  (\vec x - \vec \mu) } \\
&= \dfrac{1}{2} \log \left(  2^n\pi^n \det(A) \right) + \dfrac{1}{2} \\
&= \dfrac{1}{2} \log  \left(  2^n\pi^n e \det(A) \right) \, .
\end{align*}
\end{exmp}

Differential entropy may become negative as example \ref{diffEntGauss}
for sufficiently small $\det(A)$ proves. Even though it might become
negative, the differential entropy, as the entropy of discrete random
variables, can be interpreted as a measure of the average uncertainty
in the random variable. As shown in \cite{Cover2006}, the entropy
corresponds to the logarithm of the volume of typical outcomes, i.e., a random variable
is more tightly concentrated the smaller its entropy.

The differential entropy behaves nicely under diffeomorphic coordinate changes $\phi:S_\mathbf{\vec X} \ra \RH^n$ on the support $S_\mathbf{\vec X}= \set{\vec x = (x_1, x_2, \ldots, x_n):  \, f(\vec x) > 0} $ of the random variables $\vec X = (X_1, X_2, \ldots, X_n)$. We obtain
\begin{equation} \label{eq:entTrans}
h(\phi(\mathbf{X})) = h(\mathbf{X}) + \int f \log |\det J_\phi| 
\end{equation}
where $J_\phi$ is the Jacobian of $\phi$. If the random variables $\vec X, \vec Y$ have a joint density function $f( \vec x, \vec y)$ and conditional density function $g( \vec x | \vec y)$, respectively, we can define the \textit{conditional entropy} $h(\vec X| \vec Y)$ as
\[
h(\vec X| \vec Y) = - \int f \log g \, .
\]
Since in general $f(\vec x| \vec y) = f(\vec x, \vec y)/f(\vec y)$, we can also write
\[
h(\vec X| \vec Y) = h(\vec X,\vec Y)-h(\vec Y) \, .
\]
But we must be careful if any of the differential entropies are infinite. $h(\vec X| \vec Y)$ is a measure of the average uncertainty of the random variable $X$ conditional on the knowledge of another random variable $Y$. \\
The \textit{relative entropy} (or \textit{Kullback Leibler divergence}) $D(f||g)$ between two densities $f$ and $g$ is defined by 
\begin{equation} \label{eq:relEnt}
D(f||g) = \int f \log \dfrac{f}{g} \, .
\end{equation}
Note that $D(f||g)$ is finite only if the support set $\set{\vec x : f(\vec x) > 0}$ of $f$ is contained in the support of $g$ (Motivated by continuity, we set $0 \log (0/0) = 0$.). While $D(f||g)$ is in general not symmetric, it is often considered as a kind of distance between $f$ and $g$. This is mainly due to its property that $D(f||g) \geq 0$ with equality if and only if $f = g$.\\
The \textit{mutual information} $I(\vec X : \vec Y)$ between two random variables $\vec X$ and $\vec Y$ with joint density $f(\vec x, \vec y)$ and respective marginal densities $g(\vec x)$ and $k(\vec y)$ is defined as 
\begin{eqnarray} \label{eq:kullback}
I(\vec X: \vec Y) & = & h( \vec X) - h(\vec X| \vec Y) \nn \\
& = & h(\vec Y) - h(\vec Y|\vec X) \nn \\
& = & h(\vec X) + h(\vec Y) - h(\vec X, \vec Y) \nn \\
& = & \int f \log \dfrac{f}{g k} \nn\\
& = & D(f || g k) \, .
\end{eqnarray}
Again, we must be careful if any of the differential entropies are infinite and the mutual information might or might not diverge in this case. \\
From the definition it is clear that the mutual information is symmetric
\[
I(\vec X:\vec Y) = I(\vec Y: \vec X) \, ,
\]
non-negative
\[
I(\vec X: \vec Y) \geq 0 
\]
and equal $0$ if and only if $\vec X$ is independent from $\vec Y$. Thus, mutual information can be considered as a general measure of statistical dependence as it detects any deviations from independence. Note, that in the case of independence, knowledge of $ \vec X$ does not reduce our uncertainty about $  \vec Y$, i.e., $ \vec X$ provides no information about $\vec Y$, and vice versa. \\
Like mutual information, \textit{conditional mutual information} $I(\vec X : \vec Y| \vec Z) $ between three random variables $\vec X, \vec Y$ and $ \vec Z$ can be written in terms of conditional entropies as
\[
I(\vec X : \vec Y|\vec Z) = h(\vec X| \vec Z) - h(\vec X|\vec Y, \vec Z)
\]  
assuming the differential entropies exist. This can be rewritten to show its relationship to mutual 
information
\[
I(\vec X: \vec Y| \vec Z) = I(\vec X: \vec Y, \vec Z) - I(\vec X: \vec Z)
\]
usually rearranged as the \textit{chain rule of mutual information} \cite{Cover2006}
\begin{equation} \label{eq:chain_rule}
I(\vec X: \vec Y, \vec Z) = I(\vec X:\vec Y|\vec Z)+I(\vec X:\vec Z) \, . 
\end{equation}

\begin{defn}\label{MarkovChain}
Suppose $\vec X, \vec Y$ and $\vec Z$ are three random variables. We call
\[
\vec X \ra \vec Y \ra \vec Z
\]
a \textit{Markov Chain} also written as
\[
\vec X \independent \vec Z \mid \vec Y
\]
if 
\[
I(\vec X: \vec Z| \vec Y) = 0 \, ,
\]
i.e., the random variable $\vec X$ is independent from $\vec Z$ given $\vec Y$.
\end{defn}

From the chain rule \eq{chain_rule} we obtain the \textit{Data Processing Inequality}

\begin{cor} \label{dataProcessing}
If $\vec X \ra \vec Y \ra \vec Z$ is a Markov Chain, then
\[
I(\vec Z: \vec X) \leq I(\vec Y: \vec X) 
\]
with identity if and only if $\vec X \ra \vec Z \ra \vec Y$.
\end{cor}

In contrast to differential entropy, mutual information and conditional mutual information are scaling invariant, that is, for three diffeomorphic maps $\phi_x, \phi_y, \phi_z$ we have \cite{Kraskov2004}
\begin{align}\label{eq:scaling}
I(\phi_x(\vec X) : \phi_y(\vec Y)&) = I(\vec X:\vec Y) \nn \\
I(\phi_x(\vec X) : \phi_y(\vec Y)&| \phi_z(\vec Z)) = I(\vec X: \vec Y|\vec Z). 
\end{align}

\begin{exmp} \label{MIGauss} Let $\vec X$ and $\vec Y$ be jointly
  normally distributed with mean $\vec \mu$, marginal variances $\sg_X^2,
  \sg_Y^2$ and covariance $\sg_{XY}$. We compute
  \begin{eqnarray*}
    I(\vec X : \vec Y) & = & h(\vec X) + h(\vec Y) - h(\vec X, \vec Y) \\
    & = & \dfrac{1}{2} \log \left( 2\pi e \sg_X^2 \right) +
    \dfrac{1}{2} \log \left( 2\pi e \sg_Y^2 \right) - \left[
      \dfrac{2}{2} \log \left( 2\pi e \right) + \frac{1}{2} \log
      \left( \sg_X^2 \sg_X^2 - \sg_{XY}^2 \right) \right] \\
    & = & \frac{1}{2} \log \left( \sg_X^2 \sg_Y^2 \right) -
    \frac{1}{2} \log \left( \sg_X^2 \sg_X^2 (1 - \rho_{XY}^2) \right) \\
    & = & - \frac{1}{2} \log \left( 1 - \rho_{XY}^2 \right)    
  \end{eqnarray*}
  where $\rho_{XY} = \frac{\sg_{XY}}{\sg_X \sg_Y}$ denotes the
  correlation coefficient between $\vec X$ and $\vec Y$.

  Similarly, we can ask how much information is required to estimate a
  random quantity $\vec X$ up to some precision. Formally, let us
  assume that $\vec X$ has a normal distribution with variance
  $\sg^2$. Now, having obtained a measurement the residual variance,
  i.e. the measurement error, is reduced to $\sg_M^2$. The benchmark
  is then provided by the following question: How much information do
  we need s.t. we feel confident of knowing $\vec X$ up to $m$
  decimals?  To this end, in the present setting we would require that
  $\sg_M \approx \frac{1}{4} 10^{-m}$ motivated by the fact that $\vec
  X$ would be distributed around the measured value $\hat{x}$ with
  95\% confidence interval $\approx [\hat{x} - 2 \sg_M, \hat{x} + 2 \sg_M]$.

  Correspondingly the required information is given by
  \[ h(\vec X) - h(\vec X_{measured}) = \frac{1}{2} \log \left( 2 \pi
    e \sg^2 \right) - \frac{1}{2} \log \left( 2 \pi e \sg_M^2 \right)
  = \log \left( \frac{\sg^2}{\sg_M^2} \right) \] Thus, in the case of
  normal distributions, the required information is a non-linear
  function of the required reduction in uncertainty as measured by the
  ratio of standard deviations $\frac{\sg}{\sg_M} = 4 \cdot 10^m \sg$.
\end{exmp}

\subsection{Fokker-Planck Equation}
We start with the $n + 1$ dimensional stochastic differential equation (SDE)
\begin{align} \label{eq:stochVol}
d S_t &= \mu(\vec v_t) S_t dt + f(\vec v_t) S_t dW^0_t \nn \\
d \vec v_t &= \b(\vec{v}_t) dt + \g(\vec v_t) d \vec{W}_t
\end{align}
$S_t$ denotes the stock price at time $t$ and $\vec{v}_t$ is the volatility process driving the stock. Since we do not want to deal with regularity issues $\b$, $\g$ and $f$ are smooth. Furthermore, we assume uniform ellipticity which implies that $\g(\vec v_t)$ is invertible for all $|\vec v_t|^2 > 0$. $W^0_t$ and $\vec{W}_t = (W^1_t, \ldots, W^n_t)$ are $1$ dimensional and $n$ dimensional standard Wiener processes, respectively, which may be coupled, that is, we allow for $d W^0_t d\vec{W}_t = (\rho_1, \ldots, \rho_n) dt$ with $\rho_i \neq 0$ for $i = 1, \ldots, n$. The second SDE of \eq{stochVol} corresponds to the linear Fokker-Planck equation
\begin{equation} \label{eq:general-fokker}
\partial_t p_t(\vec{v}) = -\nabla ( \b(\vec{v}) p_t(\vec{v})) ) + \dfrac{1}{2}\sum_{i,j=1}^{n} \partial_i\partial_j\left( \left(\gamma(\vec{v})\gamma^T(\vec{v})\right)_{ij} p_t(\vec{v}) \right) \, 
\end{equation}
for the probability density where $\partial_i = d/dv_i$ and $\nabla = (\partial_1, \ldots, \partial_n)$. We call a solution $p_t$ of the Fokker-Planck equation \textit{stationary} if
\[
\partial_t p_t(\vec{v}) = 0\, ,
\]
that is, $p_t$ does no longer depend on $t$ and we write $p_t = \rho$. We call the $n$ dimensional volatility process a \textit{gradient flow} perturbed by noise whose strength is $1$ if the SDE reads
\begin{equation} \label{eq:gradientFlow}
d  \vec{v}_t  = \sqrt{2} d  W_t - \nabla V( \vec{v}_t) dt  \, 
\end{equation}
with a smooth potential $V$. \eq{gradientFlow} corresponds to the linear Fokker-Planck
equation
\begin{equation} \label{eq:fokker}
\partial_t p = \nabla \cdot (\nabla p + p \nabla V) \, .
\end{equation}  
It is not possible to compute the time dependent solution for an arbitrary potential. However, $e^{-V}$ is a solution of the ordinary differential equation
\[
 0 = \nabla \cdot (\nabla p + p \nabla V)
\]
and if $e^{-V}$ is integrable we can define the \textit{Gibbs distribution}
\begin{equation} \label{eq:gibbs}
\rho(x) = \dfrac{1}{Z} e^{-V}
\end{equation}
where the normalization constant $Z$ is the \textit{partition function}
\[
Z = \int e^{-V}.
\]
In general, integrability of $e^{-V}$ does not guarantee the existence of a stationary solution. We have a closer look on existence and uniqueness of such a stationary solution in the one dimensional case. We define the probability current
\[
S =  - \partial_v p - p \partial_v V
\]
for a solution $p$ of \eq{fokker}. For a stationary solution $\rho$ the probability current must be constant. Thus, if it vanishes at some point $v$ it must be zero for any $v$. If this is the case, we can immediately integrate the expression for the probability current which yields the expression \eq{gibbs}. For a time dependent solution $p$ of \eq{fokker} the probability current may be written in the form
\[
S = - e^{-V}\partial_v \left( e^{V}p \right)
\]
In the stationary case, where $S$ is constant, we thus have for arbitrary $S$
\[
\rho = \dfrac{1}{Z} e^{-V} - S e^{-V} \int e^{V}
\]
One of the integration constants is determined by the normalization 
\[
\int \rho = 1 \, .
\]
The other constant must be determined from the boundary conditions, so the problem arises which boundary conditions must be used. Feller \cite{Feller1951} classified all boundary conditions in dimension one. Let denote with $x_{\text{min}}$ and $x_{\text{max}}$ the boundary points of the domain of the stationary distribution $\rho$ (which includes also the cases $x_{\text{min}} = -\infty$ and $x_{\text{max}} = \infty$). The existence of such a stationary solution forces the probability current $S$ to vanish at $x_{\text{min}}$ and $x_{\text{max}}$ and therefore $S = 0$ at any point, that is, $x_{\text{min}}$ and $x_{\text{max}}$ need to be \textit{reflecting boundaries}. This imposes the \textit{natural boundary conditions} $p(x_{\text{min}}) = p(x_{\text{max}}) = 0 $ where $p$ denotes a solution of the gradient Flow \eq{fokker}. In this case, any time dependent solution converges to the stationary solution \eq{gibbs} which proves uniqueness. We refer to  chapter $5$ in \cite{Risken} for more details. We summarize these insights in

\begin{lem} \label{existenceUniqueness}
If \eq{fokker} has a stationary solution $\rho$ then 
\begin{enumerate}
\item $S \equiv 0$
\item $e^{-V}$ is integrable
\item $\rho $ is unique and concurs with the Gibbs distribution \eq{gibbs}
\item $\rho(x_{\text{min}}) = \rho(x_{\text{max}}) = 0$ \, .
\end{enumerate}
\end{lem}

We return to the $n$ dimensional case. If
\begin{equation} \label{eq:bound}
\nabla^2 V \geq \lambda I_{n \times n} \quad \forall v \in \RH^n \text{ and }\l > 0,
\end{equation}
that is, the Hessian of $V$ is uniformly bounded away from zero by a positive real $\l$, from \cite{Bakry} follows that $\rho = e^{-V - \log(Z)}$ fulfils the logarithmic Sobolev inequality. That is, for any  smooth $g \in L^2(\RH^n, \rho)$ (the space of square integrable functions w.r.t. the measure $\rho$)
\begin{equation} \label{eq:logSobolev}
\int g^2 \log g^2 d \rho - \int g^2  d \rho  \log \left( \int g^2 \rho  d \rho  \right)  \leq \dfrac{2}{\l} \int |\nabla g|^2 d \rho
\end{equation}   
holds. If $g^2$ is normalized,  \eq{logSobolev} can be rewritten as 
\begin{equation} \label{eq:GenStamGross}
D( g^2 || \rho ) \leq \dfrac{1}{2 \l} I(g^2 || \rho)
\end{equation}
with the Kullback-Leibler divergence $D$ \eq{relEnt} and the relative Fisher-Information
\[
I(g^2 || \rho) =  \int \nabla \left| \log \dfrac{g^2}{\rho} \right|^2 g^2 \, .
\]
The logarithmic Sobolev inequality generalizes the Stam-Gross inequality \cite{Stam,Gross} where $\rho$ is assumed to be normally distributed. The logarithmic Sobolev inequality \eq{logSobolev} holds true for any distribution $\rho$, regardless whether it is a stationary solution of \eq{fokker}, as long as the Hessian of $-\log \rho$ fulfils the bound \eq{bound}. \\
If $\rho_t$ is a solution of the Fokker-Planck equation \eq{fokker} with initial distribution $\rho_0$ s.t. $D(\rho_0|| \rho) < \infty$, we have
\[
\dfrac{d}{dt} D(\rho_t || \rho) = - \int I(\rho_t  ||\rho)
\]
which yields in combination with \eq{GenStamGross} and Gronwall's inequality
\begin{equation} \label{eq:trend}
D(\rho_t || \rho) \leq e^{-2\l t} D(\rho_0 || \rho) 
\end{equation}
which proves an exponentially fast trend of any solution of \eq{fokker} towards equilibrium $\rho$. See \cite{Villani} and the references therein for more details and proofs.

\section{Inferring Volatility}
\label{sec:StochVolInfo}

We define and compute the information content of returns about the volatility for stocks subject
to the dynamics \eq{stochVol}. In the sequel we assume

\begin{ass}
The Fokker-Planck equation \eq{general-fokker} has a unique stationary solution $\rho$.
\end{ass}

\subsection{Information calculations in stochastic volatility models  \label{infoVol}}

The Euler approximation of the SDE \eq{stochVol} reads
\begin{eqnarray*}
  S_{t+\tau} - S_t & = & \mu(\vec v_t) S_t \tau + f(\vec v_t) S_t \sqrt{\tau} \epsilon^{0} \\
  \vec v_{t+\tau} - \vec v_t & = & \b(\vec v_t) \tau + \g(\vec v_t) \sqrt{\tau} \vec \epsilon
\end{eqnarray*}
where $\e^{0}$ and $\vec \e$ are $1$ and $n$ dimensional standard normal distributions s.t.
$\E[\e^0\vec \e]=(\rho_1, \ldots, \rho_n) $ with the correlation coefficients $\rho_i$ of \eq{stochVol}.
Defining $r_{t+\tau} = \frac{S_{t+\tau} - S_t}{S_t}$ the result is an independent return process where 
each $r_{t + \tau} $ is normally distributed with mean $\mu(\vec v_t) \tau$ and standard deviation 
$f(\vec v_t) \sqrt{\tau}$, that is,
\begin{eqnarray}\label{eq:eulerApp}
 r_{t+\tau} & = & \mu(\vec v_t) \tau + f(\vec v_t) \sqrt{\tau} \e^{0}  \nn \\
\vec v_{t+\tau} - \vec v_t & = & \b(\vec v_t) \tau + \g(\vec v_t) \sqrt{\tau} \vec \epsilon \, . 
\end{eqnarray}
Thus, $(r_{t+\tau}, \vec v_{t+\tau})$ are jointly generated from $\vec v_{t}$ and this process is Markovian. It follows, for instance, that the past trajectory
\[
\{r,\vec v\}_{t - m \tau} = (r_{t- m\tau}, \vec v_{t - m\tau}), \ldots, (r_{t-\tau}, \vec v_{t-\tau} )\, ,
\]
and the future one
\[
 \{r,\vec v\}_{t + n \tau} = (r_{t + \tau}, \vec v_{t + \tau}), \ldots, (r_{t + n\tau}, \vec v_{t + n\tau})
\]
provided $\vec v_t$ form a Markov Chain \ref{MarkovChain}
\[
 \{r,\vec v\}_{t - m \tau} \ra \vec v_t \ra  \{r,\vec v\}_{t + n \tau} \, .
\]
Note, this remains true if the process $\vec v_t$ is not approximated. Accordingly, the Euler approximated 
return process is a hidden Markov process where the observation $r_{t+\tau}$ is drawn depending on $\vec v_t$ and 
due to the leverage effect $ \vec v_{t + \tau} $. \fig{euler} illustrates the result of the approximation.
\begin{figure}[h]
  \begin{center}
    \begin{tikzcd}
      {} & & r_{(n-1)\tau} \arrow[dash,dashed]{d} & r_{n\tau} \arrow[dash,dashed]{d} &  r_{(n + 1)\tau} \arrow[dash,dashed]{d} & \\
     \ldots \arrow{r} & \vec v_{(n-2)\tau} \arrow{r}\arrow{ru} & \vec v_{(n-1)\tau} \arrow{r}\arrow{ru}  & \vec v_{n \tau} \arrow{r}\arrow{ru}  & \vec v_{(n + 1)\tau} \arrow{r} & \ldots 
    \end{tikzcd}
  \end{center}
  \caption{\label{fig:euler} Euler approximation of the SDE \eq{stochVol}.}
\end{figure}
Using these conditional independences, we can approximate the information 
\[  
I(f(\vec v_t) \sqrt{\tau}: \vec r_{t - n \tau}^t ) 
\]
that historical returns 
\begin{equation} \label{eq:obsRet}
\vec r_{t - n \tau}^t = r_{t-n \tau},\ldots,r_t
\end{equation}
provide about the volatility $f(\vec v_t)\sqrt{\tau}$

\begin{prop} \label{volRet}
\[
 I(f(\vec v_t) \sqrt{\tau}: \vec r_{t - n \tau}^t )   \leq  I(\vec v_t : \vec v_{t-\tau}) + I(\vec v_t : r_t | \vec v_{t-\tau})
 \]
\end{prop}

\begin{proof}
From scaling invariance of the mutual information \eq{scaling} follows
\[
I(f(\vec v_t) \sqrt{\tau}: \vec r_{t - n \tau}^t )  = I(f(\vec v_t) : \vec r_{t - n \tau}^t ) \, .
\]
Observe that
\[
r_{t - n \tau}^t  \ra \vec v_t  \ra f(\vec v_t)
\]
is a Markov Chain and the Data processing inequality \ref{dataProcessing} yields
\[
I(f(\vec v_t) : \vec r_{t - n \tau}^t ) \leq I(\vec v_t : \vec r_{t - n \tau}^t ) \, .
\]
Furthermore,
\begin{eqnarray*}
  I(\vec v_t : \vec r_{t - n \tau}^t ) & \leq & I(\vec v_t : \vec v_{t-\tau} \vec r_{t - n \tau}^t) \\
  & = & I(\vec v_t : \vec v_{t-\tau}) + I(\vec v_t : \vec r_{t - n \tau}^t | \vec v_{t-\tau}) \\
  & = & I(\vec v_t : \vec v_{t-\tau}) + \underbrace{I(\vec v_t : \vec r_{t - n \tau}^{t -\tau} | \vec v_{t-\tau})}_{= 0} + I(\vec v_t : r_t | \vec v_{t-\tau}, \vec r_{t - n \tau}^{t-\tau}) \\
  & = & I(\vec v_t : \vec v_{t-\tau}) + I(\vec v_t : r_t | \vec v_{t-\tau})
\end{eqnarray*}
by conditional independence $\vec v_t \independent \vec r_{t-n \tau}^{t-\tau} \mid \vec v_{t-\tau}$ and $(\vec v_t,r_t) \independent \vec r_{t-n \tau}^{t-\tau} \mid \vec v_{t-\tau}$.
\end{proof}

An analogous argument and computation yields also an upper bound for the information past returns
provide about future ones.

\begin{prop} \label{PastFuture}
\[
I(r_{t + \tau} : \vec r_{t - n \tau}^t)  \leq I(r_{t + \tau} :  \vec v_t)
\]
\end{prop}

\begin{proof}
\begin{align*}
I(r_{t + \tau} : \vec r_{t - n \tau}^t) &\leq I(r_{t + \tau} : \vec r_{t - n \tau}^t, \vec v_t) \\
&=  I(r_{t + \tau} :  \vec v_t) +  I(r_{t + \tau} : \vec r_{t - n \tau}^t | \vec v_t) \\
&= I(r_{t + \tau} :  \vec v_t) \, .
\end{align*}
\end{proof}

Hence, for computing upper bounds on the mutual informations $ I(\vec v_t : \vec r_{t - n \tau}^t ) $ and $I(r_{t + \tau} : \vec r_{t - n \tau}^t)$ it suffices to compute
\begin{align} \label{eq:mutualInfo}
I(\vec v_\tau &: \vec v_{0}) \nn \\ 
I(\vec v_\tau &: r_\tau | \vec v_0) \nn  \\
 I(r_{ \tau} &:  \vec v_0) 
\end{align}
where we have set w.l.o.g. $t = \tau$ in the first two cases and $t = 0$ in the last case. Since for continuous 
random variables infinitely much information can be obtained (indeed, the first two terms in \eq{mutualInfo} diverge to $\infty$
if $\tau \ra 0$), evaluation whether much information is gained or not from return observations needs
benchmarks. Information gain may be interpreted as uncertainty reduction. Our initial uncertainty about a future
return $r_\tau$ or volatility $f(\vec v_\tau) \sqrt{\tau}$ is quantified by their entropies. Prior to any observation, we think of $r_\tau$ being roughly normally distributed with standard deviation 
\[
\sqrt{\tau} \int f \, d\rho \equiv \sqrt{\tau}  \sg_f 
\]
i.e, the expectation of $\vec v \mapsto f(\vec v)$ w.r.t the stationary distribution $\rho$ times $\sqrt{\tau}$. Therefore, according to example \ref{diffEntGauss}, the entropy of the random variable $r_\tau$ reads
\[
h(r_\tau) = \dfrac{1}{2} \log( 2 \pi e \sg_f^2 \tau) \, .
\]
Since $\vec v_\tau$ being drawn from the stationary distribution $\rho$ with entropy $h(\vec v_\tau)$ the volatility $f(\vec v_\tau) \sqrt{\tau}$  is stationary distributed with entropy $h(f(\vec v_\tau) \sqrt{\tau})$. If $f$ is a diffeomorphic coordinate change of a one-dimensional process $v_\tau$ the entropy of the volatility is  
\begin{equation*}
h(f( v_\tau) \sqrt{\tau}) = h(v_\tau) + \int \log | f'(v) \sqrt{\tau} | d\rho(v) \, .
\end{equation*}
Similar to the example \ref{MIGauss} we can compute how much information would be required in order to know $r_\tau$ and $f(\vec v_\tau) \sqrt{\tau}$ up to some precision. We think of $r_\tau$ being quoted in percent with precision $0.01$ on a daily basis whereas $f(\vec v_\tau) \sqrt{\tau}$ is quoted like the VIX, the volatility index of the S{\&}P 500. The VIX is quoted in percentage points as well up to precision $0.01$ and translates, roughly, to the expected movement (with the assumption of a 68{\%} likelihood, i.e., one standard deviation) in the S{\&}P 500 index over the next 30-day period, which is then annualized. This fits into our setting as follows. Assume we have an annual time scale for the SDE \eq{stochVol} and read off daily returns (a situation we encounter in subsection \ref{sixModels}). Then $\tau = 1/252$ (there are 252 trading days a year) and the volatility of the daily returns in the Euler approximation is $f(\vec v_\tau) \sqrt{\tau}$. However,  since volatility is quoted on an annual basis, we are rather interested in the random variable $f(\vec v_\tau)$ than $f(\vec v_\tau) \sqrt{\tau}$. According to example \ref{MIGauss}, we have to compute the differences
\begin{align}  \label{eq:infoGap}
G_r & \equiv h(r_\tau) - \dfrac{1}{2} \log \left(  2 \pi e \sg_M^2 \right)  = \log \left( \dfrac{\sg_f\sqrt{\tau}}{\sg_M} \right) \nn \\
G_{f(\vec v)} & \equiv h(f(\vec v_\tau)) - \dfrac{1}{2} \log \left(  2 \pi e \sg_M^2 \right) 
\end{align}
where $\sg_M = \frac{1}{4}10^{-4}$ because we want precision up to $0.01$ percent. If we are in the one-dimensional setting with a diffeomorphic function $f$ then the second identity reads
\[
G_{f(v)} =  h(v_\tau) + \int \log | f'| d\rho  - \dfrac{1}{2}\log \left(  2 \pi e \sg_M^2 \right) \, .
\]
However, if the parameters of the SDE \eq{stochVol} are quoted according to a daily time scale, we have $\tau = 1$ and we have to rescale the volatility $f(\vec v_\tau)$, which is now the one on a daily basis, to $f(\vec v_\tau) \sqrt{252}$ in order to make the result comparable with the VIX. Hence, we obtain the equations
\begin{align}  \label{eq:infoGapDaily}
G_r & \equiv h(r_\tau) - \dfrac{1}{2} \log \left(  2 \pi e \sg_M^2 \right)  = \log \left( \dfrac{ \sg_f}{\sg_M} \right) \nn \\
G_{f(\vec v)} & \equiv h(f(\vec v_\tau) ) + \dfrac{1}{2} \log( 252) - \dfrac{1}{2} \log \left(  2 \pi e \sg_M^2 \right) 
\end{align}
with $\sg_M = \frac{1}{4}10^{-4}$.

\subsection{General Solution}

We derive upper bounds and proxies for the mutual informations \eq{mutualInfo} in the subsequent theorems.

\begin{defn}\label{proxy}
We define the mutual information proxy
\[
\I(\vec v_\tau, \vec v_0) = h(\vec v_0) -  \dfrac{1}{2} \int \log \left( 2^n \pi^n e \det\left(\gamma(\vec{v_0})\gamma^T(\vec{v_0})\right) \tau^n \right) d \rho(\vec v_0)
\]
where 
\[
h(\vec v_0) = - \int \log \rho(\vec v_0) \, d  \rho(\vec v_0)
\]
denotes the entropy of the stationary distribution.
\end{defn}

\begin{thm} \label{oneTwo}
In the Euler approximation scheme \fig{euler} we obtain 
\[
 I(\vec v_\tau : \vec v_0)  = \I(\vec v_\tau, \vec v_0) + \mathcal{O}(\tau^2)
\]
and
\[
I(r_\tau : \vec v_\tau| \vec v_0) =  - \dfrac{1}{2}\log \left(  \dfrac{1}{n^2} \sum_{i=1}^{n}1 - \rho_i^2  \right) 
\]
for $i = 1, \ldots, n$ where $\rho_i = dW^i_t dW^0_t$ are the correlations in the SDE \eq{stochVol}.
\end{thm}

\begin{proof}
We introduce the Fokker-Planck operator
\[
L_{FP}(\vec v) = - \nabla \b(\vec v) + \dfrac{1}{2} \sum_{i,j=1}^{n}\partial_i \partial_j \g^2(\vec v) \, .
\]
Then, the normal distribution obtained from the Euler approximation can be also rewritten as
\[
p(\vec v_\tau| \vec v_0) = e^{\tau L_{FP}(\vec v_0) } \delta_{\vec v_{\tau} - \vec v_0} = \left( 1 + \tau  L_{FP}(\vec v_0)  + \O(\tau^2) \right) \delta_{\vec v_\tau - \vec v_0}
\]
where $\delta_{x}$ denotes the delta distribution at $x$. This easily follows form considering the characteristic functions on both sides, but see also chapter 4 in \cite{Risken} for a proof. The Kramer-Moyal Forward Expansion for the actual transition probability $q(\vec v_\tau| \vec v_0) $, i.e., the solution of the Fokker-Planck equation \eq{general-fokker} with initial condition $\vec v_0$ for $\tau = 0$, is as well
\[
q(\vec v_\tau| \vec v_0) = \left( 1 + \tau  L_{FP}(\vec v_0)  + \O(\tau^2) \right) \delta_{\vec v_\tau - \vec v_0} \, .
\]
Hence
\begin{align*}
p(\vec v_\tau| \vec v_0) \rho( \vec v_0) &= \left( p(\vec v_\tau| \vec v_0)  - q(\vec v_\tau| \vec v_0)\right) \rho( \vec v_0) + q(\vec v_\tau | \vec v_0) \rho (\vec v_0)  \\
&= \O(\tau^2) \delta_{\vec v_\tau - \vec v_0}\rho( \vec v_0) + q(\vec v_\tau| \vec v_0) \rho( \vec v_0) \\
&= \O(\tau^2) \rho( \vec v_\tau) + q(\vec v_\tau| \vec v_0) \rho( \vec v_0) \, .
\end{align*}
Since the Fokker-Planck equation \eq{general-fokker} is linear, any superposition of solutions is a solution as well.  Thus,
\[
q(\vec v_\tau) = \int 	q (\vec v_\tau | \vec v_0) d \rho(\vec v_0) 
\]
is the solution with initial condition $\rho$, i.e., the stationary solution of \eq{general-fokker} and therefore stationary as well and we obtain
\[
p(\vec v_\tau) = \int p(\vec v_\tau| \vec v_0) \, d \rho( \vec v_0)  = \int \left(\O(\tau^2) \rho( \vec v_\tau) + q(\vec v_\tau| \vec v_0)\right) \, d \rho( \vec v_0) = \left( 1 + \O(\tau^2) \right) \rho(\vec v_\tau) \, . 
\]
Therefore, the transformation rule \eq{entTrans} for the entropy yields
\[
- \int \log p(\vec v_\tau) \, d p(\vec v_\tau) = -\int \log \rho(\vec v_\tau) \, d\rho( \vec v_\tau) + \int \log \left( 1 + \O(\tau^2) \right)  \, d\rho( \vec v_\tau) = h(\vec v_0) + \O(\tau^2)  
\]
Hence, from example \ref{diffEntGauss} and the fact that $p(\vec v_\tau| \vec v_0)$ is normally distributed with covariance matrix $\gamma(\vec{v})\gamma^T(\vec{v}) \tau$ we obtain
\begin{align*}
I(\vec v_\tau : \vec v_0) &= h(\vec v_\tau) - h(\vec v_\tau| \vec v_0) \\
&= h(\vec v_0) + \O(\tau^2)  + \int  p(\vec v_\tau| \vec v_0) \log p(\vec v_\tau| \vec v_0) \, d \vec v_\tau d \rho(\vec v_0) \\
&= h(\vec v_0) + \O(\tau^2)  - \dfrac{1}{2} \int \log \left( 2^n \pi^n e \det \left( \gamma(\vec{v_0})\gamma^T(\vec{v_0}) \right) \tau^n \right) d \rho(\vec v_0) \\
&=  \I(\vec v_\tau, \vec v_0)  + \O(\tau^2)
\end{align*}
The second mutual information reads
\[
I(r_\tau : \vec v_\tau| \vec v_0) = h(r_\tau | \vec v_0) - h(r_\tau | \vec v_\tau, \vec v_0) \, .
\]
$r_\tau$ is normally distributed with variance $f(\vec v_0)^2 \tau$ in the Euler approximation scheme of \eq{eulerApp}. Hence
\[
 h(r_\tau | \vec v_0) = \dfrac{1}{2}\int \log \left( 2 \pi e f(\vec v_0)^2 \tau \right) d\rho(\vec v_0) \, .
\]
Since  a stationary solution exists, the boundary $0$ is reflecting and therefore we have almost surely $|\vec v_0|^2 > 0$ and due to uniform ellipticity of \eq{general-fokker} the matrix $\g( \vec v_0)$ is invertible and we obtain
\[
\vec \e = \dfrac{1}{\sqrt{\tau}} \g(\vec v_0)^{-1} \left(\vec v_\tau - \b(\vec v_0) \tau \right)
\]
Since $\E[\e^0 \vec \e] = (\rho_1, \ldots, \rho_n) = \vec \rho$ in \eq{eulerApp}, the return reads
\begin{align*}
r_\tau &= \mu(\vec v_0) \tau + f(\vec v_0) \sqrt{\tau} \e^0 \\
&= \mu(\vec v_0) \tau + \dfrac{f(\vec v_0) \sqrt{\tau}}{n} \sum_{i=1}^{n} \e^0 \\
&= \mu(\vec v_0) \tau + \dfrac{f(\vec v_0) \sqrt{\tau}}{n} \sum_{i=1}^{n}\rho_i \e^i + \sqrt{1 - \rho_i^2} z^i \\
&= \mu(\vec v_0) \tau  + \dfrac{f(\vec v_0)\sqrt{\tau}}{n} \sum_{i=1}^{n} \sqrt{1 - \rho_i^2} z^i + \dfrac{f(\vec v_0) \sqrt{\tau}}{n} \langle \vec \rho, \vec \e \rangle \\
&= \mu(\vec v_0) \tau  + \dfrac{f(\vec v_0)\sqrt{\tau} }{n} \sum_{i=1}^{n} \sqrt{1 - \rho_i^2} z^i + \dfrac{f(\vec v_0)}{n} \langle \vec \rho, \g(\vec v_0)^{-1} \left(\vec v_\tau - \b(\vec v_0) \tau \right) \rangle
\end{align*}
where $z^i$ are mutually independent, standard normal random variables which are also independent from $\vec \e$. $\langle \cdot, \cdot \rangle$ denotes the standard scalar product. Hence, $r_\tau$ provided $\vec v_\tau$ and $\vec v_0$ is normally distributed with mean
\[
 \mu(\vec v_0) \tau + \dfrac{f(\vec v_0)}{n} \langle \vec \rho, \g(\vec v_0)^{-1} \left(\vec v_\tau - \b(\vec v_0) \tau \right) \rangle
\]
and variance
\[
\dfrac{f(\vec v_0)^2 \tau}{n^2} \sum_{i=1}^{n}1 - \rho_i^2
\]
and therefore
\[
h(r_\tau | \vec v_\tau, \vec v_0) = \dfrac{1}{2} \int \log \left( 2 \pi e \dfrac{f(\vec v_0)^2 \tau}{n^2} \sum_{i=1}^{n} 1 - \rho_i^2 \right)\, d\rho(\vec v_0) \,. 
\]
Hence, 
\begin{align*}
I(r_\tau : \vec v_\tau | \vec v_0) &= \dfrac{1}{2}\int \log \left( 2 \pi e f(\vec v_0)^2 \tau \right) d\rho(\vec v_0) - \dfrac{1}{2} \int \log \left( 2 \pi e \dfrac{f(\vec v_0)^2 \tau}{n^2} \sum_{i=1}^{n}1 - \rho_i^2 \right) d\rho(\vec v_0) \\
&= - \dfrac{1}{2} \int \log \left(  \dfrac{1}{n^2} \sum_{i=1}^{n}1 - \rho_i^2  \right) \, d\rho(\vec v_0)
= - \dfrac{1}{2} \log \left(  \dfrac{1}{n^2} \sum_{i=1}^{n}1 - \rho_i^2  \right) 
\end{align*}
\end{proof}

\begin{rem}
\begin{enumerate} \label{oneTwoUpper}
\item The proxy $ \I(\vec v_\tau, \vec v_0) $ on the mutual information $I(\vec v_\tau : \vec v_0)$ does depend on the parametrization of the vector $\vec v_\tau$. 
\item If the dimension $n$ of the volatility process is one-dimensional, and $f$ is an diffeomorphism, we can substitute $\vec v_\tau$ and $\vec v_0$ by $f(\vec v_\tau)$ and $f(\vec v_0)$ in theorem \ref{oneTwo}, respectively, due to the invariance of mutual and conditional mutual information.
\end{enumerate}
\end{rem}

In the remaining of this subsection we assume that the stationary distribution $\rho$ fulfils the logarithmic Sobolev inequality \eq{logSobolev} with parameter $\l > 0$. Furthermore, we introduce the \textit{Fisher information matrix}
\[
\I(r_\tau; \vec v_0)_{ij} = \int \partial_{\vec v_0^i} \log p(r_t| \vec v_0)  \partial_{\vec v_0^j} \log p(r_\tau| \vec v_0) 
p(r_\tau| \vec v_0) dr_\tau
\]
with $i,j = 1, \ldots, n$ for the conditional density $p(r_\tau| v_0)$ and the \text{Fisher information}
\[
\I(r_\tau; f(\vec v_0)) = \int \left. \left( \dfrac{d}{dy}\log p(r_\tau| y) \right|_{y = f(\vec v_0)} \right)^2 p(r_\tau| f(\vec v_0)) dr_\tau 
\]
for the conditional density $p(r_\tau | f(\vec v_0))$. The Fisher information can be interpreted as the information about $\vec v_0$ and $f(\vec v_0)$, respectively, that is present in the return $r_\tau$. See \cite{Cover2006} for a detailed discussion of the properties and the interpretation of the Fisher information.

\begin{thm} \label{threeMut}
Let $\rho$ fullfill \eq{logSobolev} with $\l > 0$. Then
\begin{equation} \label{eq:ineq}
I(r_\tau : \vec v_0) \leq \dfrac{1}{2\l} \int \Tr \left( \I(r_\tau; \vec v_0) \right) d\rho(\vec v_0) \, .
\end{equation}
If there is a \textit{factor}, that is, a smooth function $\tilde \mu: \RH \ra \RH$ s.t.
 \[ 
 \tilde \mu \circ f(\vec v_0) = \mu(\vec v_0) \, ,
 \]
then
\begin{align}\label{eq:ident}
\Tr \left( \I(r_t; \vec v_0) \right) &=  \I(r_\tau;f(\vec v_0)) | \nabla f(\vec v_0)|^2  \\
									&= \left((\tilde \mu' \circ f(\vec v_0) )^2\tau + 2\right) |\nabla \log f(\vec v_0) |^2   \nn
\end{align}
\end{thm}

\begin{proof}
If we set $g = \sqrt{p(r_\tau| \vec v_0)}$ in \eq{logSobolev}, we get
\begin{align*}
I(r_\tau : \vec v_0) &= H(r_\tau) - H(r_\tau | \vec v_0) \\
& = - \int p(r_\tau) \log p(r_\tau) dr_\tau + \int p(r_\tau, \vec v_0) \log p(r_\tau| \vec v_0) d\vec v_0 dr_\tau\\
& = -\int \int p(r_\tau| \vec v_0) d \rho(\vec v_0)  \log \left( \int p(r_\tau| \vec v_0) d \rho(\vec v_0) \right) + 
\int p(r_\tau| \vec v_0) \log p(r_\tau| \vec v_0) d\rho(\vec v_0) dr_\tau\\
& \leq \dfrac{2}{\l} \int |\nabla \sqrt{p(r_\tau| \vec v_0)}|^2 d\rho(\vec v_0) dr_\tau \\
& = \dfrac{1}{2\l} \int \sum_{i = 1}^{n} \dfrac{\partial_{v^i_0}p(r_\tau| \vec v_0)^2}{p(r_\tau| v_0)}  d\rho(\vec v_0) dr_\tau \\
&=  \dfrac{1}{2\l} \int \sum_{i = 1}^{n} \left( \partial_{v^i_0} \log p(r_\tau| \vec v_0) \right)^2 p(r_\tau| \vec v_0) d\rho(\vec v_0) dr_\tau \\
&= \dfrac{1}{2\l} \int \sum_{i = 1}^{n} \I(r_\tau; \vec v_0)_{ii} d\rho(\vec v_0) \\
&=  \dfrac{1}{2\l} \int \Tr \left( \I(r_\tau; \vec v_0) \right) d\rho(\vec v_0)
\end{align*}
If $\mu(\vec v_0) = \tilde \mu \circ f(\vec v_0)$ the second identity follows from
\[
p(r_\tau | \vec v_0)  = p(r_\tau | f(\vec v_0)) \, ,
\]
that is, the conditional distribution only depends on $f(\vec v_0)$, and the chain rule
\begin{align*}
\partial_{v_0^i} \log p(r_\tau| \vec v_0)  = \partial_{v_0^i} f(\vec v_0) \left. \dfrac{d}{dy} \log p(r_\tau| y) \right|_{y = f(\vec v_0)} 
\end{align*}
for $i = 1, \ldots, n$. Finally, we compute
\begin{align*}
\dfrac{d}{dy} \log p(r_\tau | y) &= \dfrac{d}{dy} \log \left( \dfrac{1}{\sqrt{2 \pi \tau }y}e^{-\frac{(r_\tau - \tilde \mu(y) \tau)^2}{2y^2 \tau}} \right) \\
&= \dfrac{d}{dy} \left( - \log y - \dfrac{(r_\tau - \tilde \mu(y) \tau)^2}{2y^2\tau } \right) \\
&=  -\dfrac{1}{y} - \dfrac{1}{2\tau} \dfrac{- 2  (r_\tau - \tilde \mu(y) \tau) \tilde \mu'(y) \tau  y^2 - 2y (r_\tau - \tilde \mu(y) \tau)^2}{y^4} \\
&=  - \dfrac{1}{y} - \dfrac{1}{\tau}  \dfrac{ (r_\tau - \tilde \mu(y) \tau) \tilde \mu'(y) \tau  y + (r_\tau - \tilde \mu(y) \tau)^2}{y^3 }
\end{align*}
which yields
\begin{align*}
& \int   \left(  \dfrac{1}{\tau} \dfrac{(r_\tau - \tilde \mu(y) \tau)\tilde \mu'(y) \tau  y + (r_\tau - \tilde \mu(y) \tau)^2}{y^3} - \dfrac{1}{y} \right)^2 \dfrac{1}{\sqrt{2\pi  \tau}y} e^{- \frac{(x - \tilde \mu(y) \tau)^2}{2y^2 \tau}} dr_\tau \\
&=  \int \left( \left(  \dfrac{1}{\tau} \dfrac{(r_\tau - \tilde \mu(y) \tau)\tilde \mu'(y) \tau  y + (r_\tau - \tilde \mu(y) \tau)^2}{y^3} \right)^2 + \right. \\
& \quad \left. - \dfrac{2}{y}  \dfrac{1}{\tau} \dfrac{(r_\tau - \tilde \mu(y) \tau)\tilde \mu'(y) \tau  y + (r_\tau - \tilde \mu(y) \tau)^2}{y^3} 
 + \dfrac{1}{y^2} \right) \dfrac{1}{\sqrt{2\pi  \tau}y} e^{- \frac{(x - \tilde \mu(y) \tau)^2}{2y^2 \tau}}  dr_\tau \\
 &=  \int \left(  \dfrac{1}{\tau^2}\dfrac{\left( (r_\tau - \tilde \mu(y) \tau)\tilde \mu'(y) \tau \right)^2}{y^4} +  \dfrac{1}{\tau^2}\dfrac{2 (r_\tau - \tilde \mu(y) \tau)^3\tilde \mu'(y) \tau }{y^5} + \dfrac{1}{\tau^2} \dfrac{(r_\tau - \tilde \mu(y) \tau)^4}{y^6}  + \right. \\
& \quad \left. - \dfrac{2}{y} \dfrac{1}{\tau} \dfrac{(r_\tau - \tilde \mu(y) \tau)\tilde \mu'(y) \tau  y + (r_\tau - \tilde \mu(y) \tau)^2}{y^3} 
 + \dfrac{1}{y^2} \right) \dfrac{1}{\sqrt{2\pi  \tau}y} e^{- \frac{(x - \tilde \mu(y) \tau)^2}{2y^2 \tau}}  dr_\tau \\
 &= \dfrac{\tilde\mu'(y)^2 \tau}{y^2}   + \dfrac{3}{y^2} - \dfrac{2}{y^2} + \dfrac{1}{y^2} \\
 &= \dfrac{\tilde\mu'(y)^2 \tau}{y^2} + \dfrac{2}{y^2}
\end{align*}
Hence, if we substitute $y$ by $f(\vec v_0)$, we obtain
\[
 \I(r_\tau;f(v_0)) = \dfrac{(\tilde \mu' \circ f(\vec v_0) )^2\tau}{f(\vec v_0)^2} + \dfrac{2}{f(\vec v_0)^2}\, 
\]
and therefore
\begin{align*}
 \I(r_\tau;f(v_0)) |\nabla f(\vec v_0) |^2 &=  \left( \dfrac{(\tilde \mu' \circ f(\vec v_0) )^2\tau}{f(\vec v_0)^2} + \dfrac{2}{f(\vec v_0)^2} \right)  |\nabla f(\vec v_0) |^2\\
&= \left( ((\tilde \mu' \circ f(\vec v_0) )^2\tau + 2 \right) |\nabla \log f(\vec v_0) |^2
\end{align*}
\end{proof}

\begin{rem} \label{remthreeMut}
\begin{enumerate}
\item If there is a constant factor $\tilde \mu$, the upper bound does no longer depend on $\tau$. Hence, previous returns $\vec r_{t - n \tau}^t$ do not provide more information about the future return $r_{t + \tau}$ than suggested by the $\tau$-independent upper bound derived in theorem \ref{threeMut} even if $\tau \ra 0$ and $n \ra \infty$.
\item The inequality \eq{ineq} and identity \eq{ident} hold true in general, i.e., for any smooth conditional density $p(r_\tau | \vec v_0)$. 
\item Inequality \eq{ineq} has a nice interpretation: the mutual information $I(r_\tau : \vec v_0)$ is bounded by the weighted Fisher information $\I(r_\tau ; \vec v_0)$ every return yields on the particular volatility $\vec v_0$ with weights $\rho(\vec v_0)$.
\item The upper bound \eq{ineq} for the mutual information  depends on the parametrization of $\vec v_0$ whereas the mutual information $I(r_\tau : \vec v_0)$ itself is invariant. 
\item Recall, the proof of theorem \ref{oneTwo} yields 
\[
p(\vec v_\tau) = \int p(\vec v_\tau | \vec v_0) \, d\rho(\vec v_0) = \rho(\vec v_\tau)
\]
where $p(\vec v_\tau | \vec v_0) $ is the solution of the Fokker-Planck euqation \eq{general-fokker} with initial condition $ \delta_{\vec v_0}$. According to \eq{kullback} the mutual information $I(\vec v_\tau : \vec v_0)$ reads
\begin{align*}
I(\vec v_\tau : \vec v_0) &= D( p(\vec v_\tau, \vec v_0) || p(\vec v_\tau) p(\vec v_0) ) \\
&= \int p(\vec v_\tau, \vec v_0) \log \dfrac{p(\vec v_\tau, \vec v_0)}{p(\vec v_0) p(\vec v_\tau)} d\vec v_\tau d\vec v_0 \\
&= \int p(\vec v_\tau| \vec v_0) \log \dfrac{p(\vec v_\tau| \vec v_0)}{\rho(\vec v_\tau) } d\vec v_\tau d \rho(\vec v_0) \\
&= \int D(\rho_\tau || \rho) d\rho \, 
\end{align*}
where $\rho_\tau = p(\vec v_\tau| \vec v_0) $ is the solution of the Fokker-Planck equation \eq{general-fokker} with initial condition $\rho_0 = \delta_{\vec v_0}$. Formally $D(\delta_{\vec v_0} || \rho) = \infty$. Informally, one could argue, volatility is not known precisely and we approximate the Dirac distribution $\delta_{\vec v_0} $ by an initial distribution $\rho_{\vec v_0}$ with mean $\vec v_0$ and small variance s.t. $D(\rho_{\vec v_0} || \rho) < \infty$. In the case that the stationary distribution fulfils the logarithmic Sobolev inequality \eq{logSobolev} we obtain from \eq{trend}
\[
I(\vec v_\tau : \vec v_0) \leq e^{-2\l \tau} \int D(\rho_{\vec v_0} || \rho) d\rho
\]
Hence, assuming that the integral on the right of the inequality is finite, we have an informal argument that the mutual information $I(\vec v_\tau : \vec v_0)$ decreases exponentially in time.
\item As in remark \ref{oneTwoUpper}, if the volatility vector $v_0$ is one dimensional and $f$ diffeomorphic, $I( r_\tau, v_0) $ can  be replaced by $I( r_\tau, f(v_0) ) $ in theorem \ref{threeMut}.
\end{enumerate}
\end{rem}

\section{Stochastic Volatility Models}
\label{sec:StochVolModels}

We compute the mutual informations \eq{mutualInfo} for various stochastic volatility models.

\subsection{Mean Reverting One Factor Models} \label{sixModels}

We follow \cite{sixFactor} and consider six stochastic volatility models of the form
\begin{align*}
dS_t &= r S_t dt + \sqrt{v_t} S_t dW^0_t \\
dv_t &= \g v_t^{a}( \t - v_t) dt + \k v_t^b dW^1_t
\end{align*}
with $dW^1_t dW^0_t = \rho dt$ where $a \in \set{0, 1}$ and $b \in \set{1/2, 1, 3/2}$. Jones \cite{Jones} proves that for $a = 0$ and $b > 1$, there are unique stationary solutions. Analogous arguments yield stationary solutions (and therefore also uniqueness, according to lemma \ref{existenceUniqueness}) in the cases $b = 1; a = 0,1$ and $b = 1/2; a=0$.  In these cases we compute the stationary distribution and upper bounds or proxies for the mutual informations \eq{mutualInfo}. If $b=1/2$ and $a=1$, we can prove that there is no stationary solution. We start with

\begin{lem} \label{sixFactorCor}
If there is a stationary distribution $\rho$ for $v$ we have
\begin{align*}
\I( v_\tau, v_0) &=  h(v_0) - \dfrac{1}{2} \log (2 e\pi \k^2 \tau) - b \int \log v_0 \, d \rho(v_0) \\
I(r_\tau : v_\tau| v_0) &= - \dfrac{1}{2} \log ( 1 - \rho^2)
\end{align*}
\end{lem}

\begin{proof}
According to theorem \ref{oneTwo} we have
\begin{align*}
\I(v_\tau, v_0) &= h(v_0) - \dfrac{1}{2} \int \log \left( 2 \pi e \k^2 v_0^{2b} \tau \right) d \rho(v_0) \\
 &= h(v_0) - \dfrac{1}{2} \log \left( 2 \pi e \k^2  \tau \right)  - b \int \log v_0 \,  d \rho(v_0) \, .
\end{align*}
and the second equality follows directly with $n = 1$.
\end{proof}

In order to compute the stationary solutions for the various models explicitly we transform the SDE for $v_t$ into the gradient flow \eq{gradientFlow} form.

\begin{lem} \label{sixFactorLemma}
Define
\[
\begin{array}{ll}
g(v) = \dfrac{\sqrt{2}}{\k(1 - b)} v^{1 - b} & \text{if } b \neq 1 \\ 
g(v) = \dfrac{\sqrt{2}}{\k} \log v & \text{if } b = 1
\end{array} 
\]
and $\sg = g(v)$. Then
\[
d \sg_t = \sqrt{2} dW^1_t + V'(\sg_t) dt
\]
with 
\[
\begin{array}{ll}
V'(\sg) = \sqrt{2} \left\lbrace -\dfrac{\g}{\k} \left( \t \left( \dfrac{\sg \k (1 -b)}{\sqrt{2}} \right)^{\frac{a - b}{1 - b}} - \left( \dfrac{\sg \k (1 -b)}{\sqrt{2}} \right)^{\frac{a - b + 1}{1 - b}} \right) + \dfrac{b }{ ( 1 - b)\sqrt{2}}\dfrac{1}{\sg} \right\rbrace & \text{if } b \neq 1 \\ 
V'(\sg) = \sqrt{2} \left\lbrace - \dfrac{\g}{\k} \left(\t e^{\frac{(a - 1)\k}{\sqrt{2}}\sg } - e^{\frac{a\k}{\sqrt{2}}\sg}\right) +\dfrac{\k}{2} \right\rbrace  & \text{if } b = 1
\end{array} 
\]
\end{lem}

\begin{proof}
For $b \neq 1$ we have
\begin{align*}
g'(v)  &= \dfrac{\sqrt{2}}{\k} \dfrac{1}{v^b} \\
g''(v) &= - \dfrac{b \sqrt{2}}{\k} \dfrac{1}{v^{1 + b}} \\
v &=\left( \dfrac{\sg \k (1-b)}{\sqrt{2}} \right)^{\frac{1}{1-b}} 
\end{align*}
and It\^{o}'s formula yields
\begin{align*}
d \sg &= g'(v_t) dv_t + \dfrac{g''(v_t)}{2} dv_tdv_t \\
&= \dfrac{\sqrt{2}}{\k} \dfrac{1}{v_t^b} \left( \g v_t^{a}( \t - v_t) dt + \k v_t^b dW^1_t \right) 
- \dfrac{b}{\k \sqrt{2}} \dfrac{1}{v_t^{1 + b}}\left( \k v_t^b \right)^2 dt \\
&= \sqrt{2} dW^1_t + \dfrac{\sqrt{2}\g}{\k}v^{a - b}(\t - v_t) dt - \dfrac{b \k}{\sqrt{2}} v^{b-1} dt \\
&= \sqrt{2} dW^1_t + \sqrt{2} \left\lbrace \dfrac{\g}{\k} \left( \dfrac{\sg \k (1-b)}{\sqrt{2}} \right)^{\frac{a - b}{1 - b}}  \left( \t - \left( \dfrac{\sg \k (1-b)}{\sqrt{2}} \right)^{\frac{1}{1-b}}\right) - \dfrac{b}{\sqrt{2}(1-b)}\dfrac{1}{\sg} \right\rbrace dt \\
&=  \sqrt{2} dW^1_t + \sqrt{2} \left\lbrace \dfrac{\g}{\k} \left( \t  \left( \dfrac{\sg \k (1-b)}{\sqrt{2}} \right)^{\frac{a - b}{1 - b}}  - \left( \dfrac{\sg \k (1-b)}{\sqrt{2}} \right)^{\frac{a - b + 1}{1-b}}\right) - \dfrac{b}{\sqrt{2}(1-b)}\dfrac{1}{\sg} \right\rbrace dt \, .
\end{align*}
For $b = 1$ we have
\begin{align*}
g'(v)  &= \dfrac{\sqrt{2}}{\k} \dfrac{1}{v} \\
g''(v) &= - \dfrac{\sqrt{2}}{\k}\dfrac{1}{v^2} \\
v &= e^{\k \sg/\sqrt{2}} \, 
\end{align*}
and It\^{o}'s formula yields
\begin{align*}
d \sg &= g'(v_t) dv_t +  \dfrac{g''(v_t)}{2} dv_tdv_t \\
&= \dfrac{\sqrt{2}}{\k} \dfrac{1}{v_t} \left( \g v_t^{a}( \t - v_t) dt + \k v_t dW^1_t \right) 
- \dfrac{1}{\k \sqrt{2}} \dfrac{1}{v_t^{2}}\left( \k v_t \right)^2 dt \\
&=  \sqrt{2} dW^1_t + \dfrac{\sqrt{2}\g}{\k}v_t^{a - 1}(\t - v_t) dt - \dfrac{\k}{\sqrt{2}}  dt \\
&=  \sqrt{2} dW^1_t + \sqrt{2} \left\lbrace \dfrac{\g}{\k} \left(\t e^{\frac{(a - 1)\k}{\sqrt{2}}\sg } - e^{\frac{a\k}{\sqrt{2}}\sg}\right) -\dfrac{\k}{2} \right\rbrace dt
\end{align*}
\end{proof}

In order to apply theorem \ref{threeMut} the following lemma turns out to be useful.

\begin{lem} \label{squareGamma}
The function
\[
\rho(\sg) = \dfrac{2 \b^\a}{\Gamma( \a)} \sg^{2 \a - 1} e^{-\b \sg^2}
\]
defines a probability density for all $\a, \b > 0$. $\sg^2$ is Gamma distributed with shape $\a$ and rate $\b$. The density $\rho$ fulfils the logarithmic Sobolev inequality \eq{logSobolev} with
\[
\l = 2\b \, .
\]
Furthermore, assume the process $v_t$ has a stationary distribution $\rho(v)$ and $v = g (\sg)$ with
\[
g(\sg) = \eta \sg^{b} \quad b \in \set{-2, 2}, \, \eta > 0\, 
\]
and $\a > 1$, then
\[
I(r_\tau : v_0)  \leq \dfrac{1}{2(\a - 1)} \, .
\]
\end{lem}

\begin{proof}
We first check that $u = h(\sg) = \sg^2$ is Gamma distributed with shape $\a$ and rate $\b$. The transformation rule for probability densities yields
\begin{align*}
\rho(u) & = \rho\left(h^{-1}(u)\right)\dfrac{1}{h' \left(h^{-1}(u) \right)} \\
&= \rho(\sqrt{u})\dfrac{1}{2\sqrt{u}} \\
&= \dfrac{2 \b^\a}{\Gamma( \a)} u^{\a - 1/2} e^{-\b u} \dfrac{1}{2\sqrt{u}} \\
&= \dfrac{ \b^\a}{\Gamma( \a)}  u^{\a - 1} e^{-\b u} \, 
\end{align*}
which is the Gamma distribution with shape $\a$ and rate $\b$. Hence, $\rho(\sg)$ is a density as well.
Furthermore, if we write $\rho(\sg) = e^{-V(\sg)}$ we obtain
\[
V(\sg) = \b \sg^2 -  (2\a - 1) \log \sg + Z
\]
for a constant $Z$. Hence
\[
V''(\sg) = 2\b
\]
is bounded from below by $2 \b > 0$ and $\rho$ fulfils the logarithmic Sobolev inequality. Since mutual information is scaling invariant we have for $\sg_0 = g^{-1}(v_0)$
\[
I(r_\tau: v_0) = I(r_\tau : \sg_0) \, .
\]
We apply theorem \ref{threeMut} and get
\begin{align}
\l &= 2\b \nn \\
f(\sg_0) &= \sqrt{g(g^{-1}(v_0))} = \sqrt{ \eta \sg_0^{b}} = \sqrt{\eta} \sg_0^{b/2}; \, b \in \set{-2, 2} \label{eq:f}\\
\tilde \mu'& = 0 \nn
\end{align}
where the last identity follows from the fact that the drift term is constant. Hence, we obtain
\begin{align*}
I(r_\tau : \sg_0)  &\leq \dfrac{1}{4 \b} \int 2 \left| \dfrac{d}{d \sg} \log f (\sg_0)  \right|^2 \, d\rho(\sg_0)\\
&= \dfrac{1}{2\b} \int\dfrac{1}{\sg_0^2} \, d\rho(\sg_0) \\
 &= \dfrac{1}{2\b} \int  \dfrac{1}{\sg_0^2} \dfrac{2 \b^\a}{\Gamma( \a)} \sg_0^{2 \a - 1} e^{-\b \sg_0^2} \,  d\sg_0 \\
 &= \dfrac{1}{2}\int \dfrac{2 \b^{\a-1}}{\Gamma( \a)} \sg_0^{2 (\a - 1)} e^{-\b \sg_0^2} \,  d\sg_0 \\
 &= \dfrac{\Gamma(\a - 1)}{2 \Gamma(\a)} \\
 &= \dfrac{1}{2(\a-1)} \, 
\end{align*}
where the second identity follows from \eq{f}.
\end{proof}

\begin{rem} \label{applicable}
In the cases $b = 1/2$ and $b = 3/2$ we have from lemma \ref{sixFactorLemma}
\[
v = \left( \dfrac{\sg \k (1-b)}{\sqrt{2}} \right)^{\frac{1}{1-b}} \, ,
\]
that is, we have a transformation rule demanded in lemma \ref{squareGamma}. Besides, for $b = 3/2$ we deal 
with $\sg < 0$ and care is necessary in the sequel computations if logarithms or roots of $\sg$ are involved.
\end{rem}

\begin{cor} \label{3/20}
Suppose $b = 3/2$ and $a = 0$. If 
\begin{equation}
3\k \sqrt{\dfrac{\t}{\g}} > 1 \, .
\end{equation}
then
\begin{align*}
I(r_\tau : v_0) &\leq \dfrac{1}{12\sqrt{ \a'} - 4 \b'} \int_{-\infty}^{0} \dfrac{1}{\sg^2} d \rho(\sg) \\
\rho(\sg_0) &= \dfrac{1}{Z}(- \sg_0)^3 e^{-\a'\left(\sg_0^2 -  \b'/\a'\right)^2}  \\
\rho(v_0) &= \dfrac{1}{Z'} \dfrac{1}{v_0^3} e^{- \a \left( 1/v_0 - 1/\t \right)^2 } \\
\I(v_\tau, v_0)&=   h(v_0) - \dfrac{1}{2} \log (2 e\pi \k^2 \tau) - \dfrac{3}{2} \int \log v_0 \, d \rho(v_0)
\end{align*}
with
\begin{align*} 
\a' &= \dfrac{\g \k^2 \t}{64}\\
\b' &= \dfrac{\g}{8} \nn \, .\\
\a &= \dfrac{\g \t}{\k^2}
\end{align*}
\end{cor}

\begin{proof}
\begin{align*}
V'(\sg) &= \sqrt{2} \left\lbrace -\dfrac{\g}{\k} \left( \t \left( - \dfrac{\sg \k }{\sqrt{8}} \right)^{3} - \left(- \dfrac{\sg \k }{\sqrt{8}} \right)\right) - \dfrac{3 }{ \sqrt{2}}\dfrac{1}{\sg} \right\rbrace \\
&=  \dfrac{\g\sqrt{2}}{\k} \left( \t \left( \dfrac{\sg \k }{\sqrt{8}} \right)^{3} -\dfrac{\sg \k }{\sqrt{8}} \right) - \dfrac{3}{\sg}\\
& = \dfrac{\g}{2} \left(  \dfrac{\t \k^2}{8} \sg^3 -  \sg \right) - \dfrac{3}{\sg} \\
&=  4 (\a' \sg^3 - \b' \sg)  - \dfrac{3}{\sg} 
\end{align*}
for 
\begin{align*} 
\a' &= \dfrac{\g \k^2 \t}{64}\\
\b' &= \dfrac{\g}{8} \, .
\end{align*}
This implies
\begin{align*}
V''(\sg) &= 4 (3 \a' \sg^2 - \b')  + \dfrac{3}{\sg^2}\\
V(\sg) &=  \a' \sg^4 -  2 \b' \sg^2 - 3 \log (-\sg) = \a'\left(\sg^2 -  \dfrac{\b'}{\a'}\right)^2 -\dfrac{\b'^2}{\a'} - 3 \log( -\sg)
\end{align*}
One can read of that $e^{-V}$ is integrable and the Gibbs distribution \eq{gibbs} reads in this case
\[
\rho(\sg) = \dfrac{1}{Z}(- \sg)^3 e^{-\a'\left(\sg^2 -  \b'/\a'\right)^2} \, .
\]
$Z$ needs to be computed numerically. We check \eq{bound} by computing the third derivative.
\[
						 0 = V'''(\sg') = 24\a' \sg' - \dfrac{6}{\sg'^3} 
\Leftrightarrow  \sg'^2 = \sqrt{\dfrac{1}{4\a'}}
\]
which yields
\begin{equation} \label{eq:lambda}
\l \equiv V''(\sg') =   4 (3 \a' \sg'^2 - \b')  + \dfrac{3}{\sg'^2} = 6 \sqrt{  \a'} - 4 \b' + 6 \sqrt{  \a'} = 12\sqrt{ \a'} - 4 \b'
\end{equation}
which is greater zero iff $3 \sqrt{ \a'}   > \b'$ that is
\[
3\k \sqrt{\dfrac{\t}{\g}} > 1 \, .
\]
If this condition holds true the logarithmic Sobolev inequality for the stationary distribution holds with 
$\l$ as in \eq{lambda} and we can apply theorem \ref{threeMut}. For the same reasons as in lemma \ref{squareGamma} we have $I(r_\tau : v_0) = I(r_\tau: \sg_0)$ and by an analogous computation as in the proof of lemma \ref{squareGamma} we obtain 
\begin{align*}
  I(r_\tau: \sg_0) \leq \dfrac{1}{12\sqrt{ \a'} - 4 \b'} \int \dfrac{1}{\sg_0^2} \, d\rho(\sg_0) \, .
\end{align*}
We compute the stationary for the variance $v$ itself. According to lemma \ref{sixFactorLemma}, we have
\[
\sg^2 = \dfrac{8}{v \k^2}
\]
and therefore
\[
2 \sg g'(v) = - \dfrac{8}{v^2 \k^2} \, .
\]
This yields the stationary distribution for $v$ 
\begin{align*}
\rho(v) &= \dfrac{1}{Z'} (-\sg)^3 e^{-\a' \left(  \frac{8}{v \k^2} - \frac{\b'}{\a'} \right)^2} \dfrac{-1}{\sg v^2}  \\
&= \dfrac{1}{Z'} \dfrac{1}{v^3} e^{- \a \left( 1/v - 1/\t \right)^2 }
\end{align*}
for a constant $Z'$ and
\[
\a = \a'  \dfrac{64}{\k^4} = \dfrac{\g \t}{\k^2} 
\]
\end{proof}

\begin{cor} \label{3/21}
Suppose $b = 3/2$ and $a = 1$. Then
\begin{align*}
I(r_\tau : v_0) &\leq \dfrac{1}{2 (\a - 1)} \\
\rho(v_0) &= \dfrac{\b^\a}{\Gamma(\a)} v_0^{- a - 1} e^{-\b/v_0} \\
\I(v_\tau, v_0)&=  \a + \dfrac{1}{2}\log \left( \dfrac{\Gamma(\a)^2}{\b} \right) - \left( \a - \dfrac{1}{2}\right) \Psi(\a) - \dfrac{1}{2} \log (2 e\pi \k^2 \tau)
\end{align*}
where $\Gamma$ is the Gamma function, $\Psi$ the Digamma function and 
\begin{align*} 
\a &= \dfrac{2\g}{\k^2} + 2  \\
\b &= \dfrac{2\g\t}{\k^2} \, .
\end{align*}
\end{cor}

\begin{proof}
\begin{align*}
V'(\sg) &= \sqrt{2}\left\lbrace - \dfrac{\g}{\k} \left( \t \left( - \dfrac{\sg \k }{\sqrt{8}}\right) - \left( - \dfrac{\sg \k }{\sqrt{8}}\right)^{-1} \right) - \dfrac{3}{\sqrt{2}}\dfrac{1}{\sg}\right\rbrace \\
&= \dfrac{\g}{\k} \left( \t  \dfrac{\sg \k }{2} -  \dfrac{4}{\sg \k} \right) - \dfrac{3}{\sg}\\
&= 2\b' \sg -  \dfrac{(2 \a -1)}{\sg}
\end{align*}
with 
\begin{align*} 
\a &= \dfrac{2\g}{\k^2} + 2  \\
\b' &= \dfrac{\g\t}{4} \, .
\end{align*}
This implies
\begin{align*}
V''(\sg) &= 2\b' +  \dfrac{(2 \a -1)}{\sg^2} > 2\b' \equiv \l \quad \text{for all }\sg \\
V(\sg) &= \b' \sg^2 - (2 \a -1)\log ( - \sg) \, 
\end{align*}
and therefore the stationary distribution $\rho$ fulfils the logarithmic Sobolev inequality \eq{logSobolev} and reads
\[
\rho(\sg) = \dfrac{1}{Z} (-\sg)^{2 \a -1}e^{-\b' \sg^2} \, .
\]
Hence, $\sg^2$ is Gamma distributed with shape $\a$ and rate $\b'$ and we recognize the partition function
\[
Z = \dfrac{\Gamma(\a) }{2\b'^\a} \, .
\]
with the Gamma function $\Gamma$. Lemma \ref{squareGamma} proves the inequality for the mutual information $I(r_\tau: v_0)$. Finally, we have
\[
\sg^2 = g(v)^2 = \dfrac{8}{\k^2}\dfrac{1}{v}
\]
and therefore the variance $v$ has the stationary distribution
\begin{align*}
\rho(v) &=  \dfrac{2\b'^\a}{\Gamma(\a)} \left( \dfrac{8}{\k^2}\dfrac{1}{v} \right) ^{\a -1/2}e^{-\b' \frac{8}{\k^2}\frac{1}{v}} \dfrac{\sqrt{2}}{\k}\dfrac{1}{v^{3/2}} \\
&= \dfrac{2\left( \b' \dfrac{8}{\k^2} \right)^\a}{\Gamma(\a)} v^{-\a + 1/2}v^{-3/2}e^{-\b' \frac{8}{\k^2}\frac{1}{v}} \\
&= \dfrac{\b^\a}{\Gamma(\a)} v^{- a - 1} e^{-\b/v} \, ,
\end{align*}
that is, an inverse Gamma distribution with rate $\a$ and shape 
\[
\b = \dfrac{8\b'}{\k^2} = \dfrac{2\g\t}{\k^2}\, 
\]
and we get
\begin{align*}
h(v) &= \a + \log \left( \b \Gamma(\a) \right) - ( 1 + \a) \Psi(\a) \\
\int \log v \, d \rho(v) &= \log \b - \Psi(\a)
\end{align*}
where $\Psi$ denotes the Digamma function. Lemma \ref{sixFactorCor} yields the expression for the proxy $\I(v_\tau, v_0)$.
\end{proof}

\begin{cor} \label{oneZero}
Suppose $b = 1$ and $a = 0$. Then
\begin{align*}
I(r_\tau : v_0) &\leq \dfrac{1}{2 (\a - 1)} \\
\rho(v_0) &= \dfrac{\b^\a}{\Gamma(\a)} v_0^{- \a - 1} e^{-\b/v_0} \\
\I(v_\tau, v_0) &=  \a (1 - \Psi(\a))+ \log \Gamma(a) - \dfrac{1}{2} \log (2 e\pi \k^2 \tau)
\end{align*}
with the Gamma function $\Gamma$, Digamma function $\Psi$ and
\begin{align*} 
\a &= \dfrac{2\g}{\k^2} + 1 \\
\b &= \dfrac{2\g\t}{\k^2}
\end{align*}
\end{cor}

\begin{proof}
\begin{align*}
V'(\sg) &= \sqrt{2} \left\lbrace - \dfrac{\g}{\k} \left(\t e^{\frac{-\k}{\sqrt{2}}\sg } -1 \right) +\dfrac{\k}{2} \right\rbrace \\
&= - \b \dfrac{\k}{\sqrt{2}} e^{\frac{-\k}{\sqrt{2}}\sg } + \dfrac{\k}{ \sqrt{2}} \left( \dfrac{2\g}{\k^2} + 1\right) \\
&= - \b  \dfrac{\k}{\sqrt{2}} e^{\frac{-\k}{\sqrt{2}}\sg } + \a \dfrac{\k}{ \sqrt{2}} 
\end{align*}
with 
\begin{align*} 
\a &= \dfrac{2\g}{\k^2} + 1  \\
\b &= \dfrac{2\g\t}{\k^2}
\end{align*}
and therefore
\begin{align*}
V''(\sg) &=   \b \dfrac{\k^2}{2} e^{\frac{-\k}{\sqrt{2}}\sg } \\
V(\sg) &= \b e^{\frac{-\k}{\sqrt{2}}\sg } +  \a \dfrac{\k}{ \sqrt{2}}  \sg 
\, .
\end{align*}
One can read off that $V''$ is not uniformly bounded away from zero, and $e^{-V}$ is integrable. If we define 
\[
\sg' = \dfrac{1}{\sqrt{v}} = e^{-\frac{\k}{\sqrt{8}}\sg}
\] 
then the stationary distribution $\rho$ \eq{gibbs} for $\sg'$ reads
\begin{align*}
\rho(\sg') &= \dfrac{1}{Z} e^{-V(\sg)}  e^{\frac{\k}{\sqrt{8}}\sg} \\
&= \dfrac{1}{Z}  e^{-\b \sg'^2} e^{-\a  \frac{\k}{ \sqrt{2}} \sg}e^{\frac{\k}{\sqrt{8}}\sg} \\
&= \dfrac{1}{Z}  e^{-\frac{\k}{\sqrt{8}}\sg (2\a - 1) } e^{-\b \sg'^2} \\
&= \dfrac{1}{Z}  \sg'^{ 2\a - 1 }e^{-\b \sg'^2}  \, .
\end{align*}
That is $\sg'^2$ is Gamma distributed with shape $\a$ and rate $\b$. Lemma \ref{squareGamma} yields
\begin{align*}
Z &= \dfrac{\Gamma(\a)}{2 \b^\a} \\
I(r_\tau:\sg'_0)  &\leq \dfrac{1}{2(\a - 1)} \\
\end{align*}
As in the case $b = 3/2, a= 1$ one checks that the variance $v$ of the Stock process is inverse Gamma distributed with shape $\a$ and rate $\b$.
\end{proof}

\begin{cor} \label{oneOne}
Suppose $b = a = 1$ and
\[
\dfrac{2\g\t}{\k^2} - 1 > 1 \, .
\]
Then
\begin{align*}
I(r_\tau : v_0) &\leq \dfrac{1}{2 (\a - 1)} \\
\rho(v_0) &= \dfrac{\b^\a}{\Gamma(\a)} v_0^{\a - 1} e^{-\b v_0} \\
\I(v_\tau, v_0) &=  \alpha(1 -\Psi(\a)) + \log \Gamma(\a) - \dfrac{1}{2} \log (2 e\pi \k^2 \tau)
\end{align*}
with the Gamma function $\Gamma$, Digamma function $\Psi$ and
\begin{align*} 
\a &= \dfrac{2\g\t}{\k^2} - 1\\
\b &= \dfrac{2\g}{\k^2}
\end{align*}
\end{cor}

\begin{proof}
\begin{align*}
V'(\sg) &= \sqrt{2}\left\lbrace -\dfrac{\g}{\k} \left( \t - e^{\frac{\k}{\sqrt{2}}\sg} \right) + \dfrac{\k}{2} \right\rbrace \\
&= \dfrac{2\g}{\k^2} \dfrac{\k}{\sqrt{2}} e^{\frac{\k}{\sqrt{2}}\sg} + \dfrac{\k}{\sqrt{2}} \left( 1 - \dfrac{2\g\t}{\k^2} \right) \\
& = \b \dfrac{\k}{\sqrt{2}} e^{\frac{\k}{\sqrt{2}}\sg} - \a \dfrac{\k}{\sqrt{2}} 
\end{align*}
with
\begin{align*}
\a &= \dfrac{2\g\t}{\k^2} - 1\\
\b &= \dfrac{2\g}{\k^2}
\end{align*}
and therefore
\begin{align*}
V''(\sg) &= \b \dfrac{\k^2}{2} e^{\frac{\k}{\sqrt{2}}\sg} \\
V(\sg) &= \b e^{\frac{\k}{\sqrt{2}}\sg}  - \a \dfrac{\k}{\sqrt{2}}  \sg \, .
\end{align*}
$V$ is not bounded from below by a positive real and $e^{-V}$ is only integrable if $\a > 0$. In this
case we introduce
\[
\sg' = \sqrt{v} = e^{\frac{\k}{\sqrt{8}}\sg}
\]
and the stationary distribution w.r.t. $\sg'$ reads
\begin{align*}
\rho(\sg') &= \dfrac{1}{Z} e^{-V(\sg)}  e^{\frac{\k}{\sqrt{8}}\sg} \\
&= \dfrac{1}{Z}  e^{-\b \sg'^2} e^{-\a  \frac{\k}{ \sqrt{2}} \sg}e^{\frac{\k}{\sqrt{8}}\sg} \\
&= \dfrac{1}{Z}  e^{-\frac{\k}{\sqrt{8}}\sg (2\a - 1) } e^{-\b \sg'^2} \\
&= \dfrac{1}{Z}  \sg'^{ 2\a - 1 }e^{-\b \sg'^2}  \, .
\end{align*}
That is, the variance $v = \sg'^2$ is Gamma distributed with shape $\a$ and rate $\b$. The origin is reflecting iff $\a >1$. For a Gamma distributed variance $v$ with shape $\a$ and rate $\b$ we obtain
\begin{align*}
h(v) &= \a - \log \b + \log \Gamma(\a)  + ( 1 - \a) \Psi(\a) \\
\int \log v \, d \rho(v) &= \Psi(\a) - \log \b
\end{align*}
The expression for the proxy $\I(v_\tau, v_0)$ then follows from lemma \ref{sixFactorCor}.
\end{proof}

\begin{cor}\label{heston}
Suppose $b = 1/2$, $a = 0$, and 
\[
\dfrac{2\g\t}{\k^2} > 1\, .
\]
Then
\begin{align*}
I(r_\tau : v_0) &\leq \dfrac{1}{2 (\a - 1)} \\
\rho(v_0) &= \dfrac{\b^\a}{\Gamma(\a)} v_0^{\a - 1} e^{-\b v_0} \\
\I(v_\tau, v_0) &=  \a - \dfrac{1}{2} \log \b + \log \Gamma(\a)  + \left( \dfrac{1}{2} - \a \right) \Psi(\a) - \dfrac{1}{2} \log (2 e\pi \k^2 \tau)
\end{align*}
with the Gamma function $\Gamma$, Digamma function $\Psi$ and
\begin{align*} 
\a &= \dfrac{2\g\t}{\k^2} \\
\b &= \dfrac{2\g}{\k^2} \
\end{align*}
\end{cor}

\begin{proof}
\begin{align*}
V'(\sg) &= \sqrt{2}\left\lbrace - \dfrac{\g}{\k} \left( \t \left( \dfrac{\sg \k}{\sqrt{8}} \right)^{-1} - \left( \dfrac{\sg \k}{\sqrt{8}} \right) \right) + \dfrac{1}{\sqrt{2}} \sg \right\rbrace \\
&= \dfrac{\g}{2} \sg - \dfrac{4 \t \g}{\k^2}\dfrac{1}{\sg} + \dfrac{1}{\sg} \\
&= 2\b' \sg + ( 1 - 2\a) \dfrac{1}{\sg}
\end{align*}
with
\begin{align*}
\a &= \dfrac{2\g\t}{\k^2} \\
\b'&= \dfrac{\g}{4} \, .
\end{align*} 
Hence,
\begin{align*}
V''(\sg) &= 2\b' - ( 1 - 2 \a) \dfrac{1}{\sg^2} \\
V(\sg) &= \b' \sg^2 + ( 1- 2 \a) \log \sg \, .
\end{align*}
$e^{-V}$ is integrable iff $\a > 1/2$. Then $\sg^2$ is Gamma distributed with shape $\a$ and rate $\b'$, i.e., 
\[
\rho(\sg) = \dfrac{2\b'^\a}{\Gamma(\a)}\sg^{2\a - 1}e^{-\b' \sg^2}
\]
and the origin is reflecting iff $\a > 1$. The variance $v$ is Gamma distributed with rate $\a$ and shape
\[
\b = \dfrac{8}{\k^2} \b' = \dfrac{2\g}{\k^2} \, \, .
\]
Everything else follows from lemma \ref{squareGamma} and lemma \ref{sixFactorCor}.
\end{proof}

\begin{cor}
Suppose $b = 1/2$ and $a = 1$. Then, the Fokker-Planck equation
\[
\partial_t p = - \partial_v(\g v(\t - v) p ) + \dfrac{1}{2} \partial_{vv} \left( \k^2 v p \right)
\]
has no stationary solution.
\end{cor}

\begin{proof}
\begin{align*}
V'(\sg) &= \sqrt{2} \left\lbrace -\dfrac{\g}{\k} \left( \t \left( \dfrac{\sg \k}{\sqrt{8}} \right) - \left(  \dfrac{\sg \k}{\sqrt{8}} \right)^{3} \right) + \dfrac{1}{\sqrt{2}}\dfrac{1}{\sg} \right\rbrace \\
&= -\dfrac{\g}{\k} \left( \t \dfrac{\sg\k}{2} - \left( \dfrac{\sg \k}{\sqrt{8}} \right)^3 \right) + \dfrac{1}{\sg} \\
&= - \dfrac{\g\t}{2} \sg + \dfrac{\g\k^2}{16} \sg^3 + \dfrac{1}{\sg}
\end{align*}
which yields 
\[
V(\sg) = -\dfrac{\g\t}{4} \sg^2 + \dfrac{\g \k^2}{64} \sg^4  + \log \sg \, 
\]
and therefore
\[
e^{-V(\sg)} = \dfrac{1}{\sg} e^{- \frac{\g \k^2}{64} \sg^4 + \frac{\g\t}{4} \sg^2}
\]
which is not an integrable function on $(0, \infty)$. The coordinate change $g$ in lemma \ref{sixFactorLemma} is a differentiable and strictly monotonic. Hence, there is a stationary solution $\rho$ if and only if 
\[
\rho(\sg) = \rho\left(g^{-1}(\sg)\right)g' \left(g^{-1}(\sg) \right) \, ,
\]
with $\sg = g(v)$, is a stationary solution of the gradient flow \eq{fokker}. Assume there is a stationary solution $\rho$ of \eq{fokker} with boundaries $x_{\text{min}}$ and $x_{\text{max}}$. If $x_{\text{max}}$ is finite, then $V(x_{\text{max}}) = \pm \infty$ -- see chapter 5 in \cite{Risken} for this point. The same holds true for $x_{\text{min}}$. Thus, the only possible boundaries are $x_{\text{min}} = 0$ and $x_{\text{max}} = \infty$. Hence, due to lemma \ref{existenceUniqueness}, the Gibbs distribution \eq{gibbs} is defined on $(0, \infty)$ and $e^{-V}$ is an integrable function on $(0, \infty)$ -- a contradiction.
\end{proof}

\begin{rem}
As far as the authors know, this is the first time that non-existence of a stationary solution for the case $b=1/2$ and $a=1$ has been proven. See also \cite{sixFactor} for a short discussion of this issue.
\end{rem}

We have been assembled all ingredients for actually computing the mutual informations \eq{mutualInfo} for the five models which have (possibly) a stationary distribution. The parameters for the models were fitted \cite{sixFactor} on Options of the S\&P 500 
\[
\begin{array}{ccccccc}
a  & b      &   \g                  & \t                   &   \k                & \l                   &     \rho \\ 
0  & 1/2 &          3.1146    & 0.0523           &  0.5826       & 9.00e-5          & -0.6520\\ 
    &        &  \pm 1.50e-3    &  \pm 3.13e-5 &  \pm 1.27e-4  &  \pm 2.16e-3  & \pm  2.41e-4\\ 
0 & 1 & 2.4730                  & 0.0272          &  1.1884          & 4.37e-3          & -0.7116 \\ 
 &     &  \pm 1.12e-3         &  \pm 6.84e-6 & \pm 3.20e-4    &   \pm 1.17e-3&  \pm  2.06e-4 \\ 
1 & 1 & 64.4378               & 0.0367          &  1.1214            & 1.93e-3        & -0.6749\\ 
 &  &   \pm 3.95e-2       &  \pm 1.91e-5 &  \pm 4.36e-4    & \pm  4.94e-2 & \pm  2.50e-4 \\ 
0 & 3/2 & 1.5384 & 0.0336  &    7.9501     & 7.91e-4& -0.7169\\ 
 &  &  \pm 2.22e-3     & \pm 3.65-5 &    \pm 2.94e-3     &  \pm 2.15e-3& \pm  2.10e-4\\ 
1 & 3/2 & 50.9140 & 0.0388 &   6.2593  & 3.36e-4 & -0.6854 \\
 & &  \pm 4.08e-2 &\pm 2.86e-5 &  \pm 2.41e-3    & \pm 5.13e-2& \pm 2.35e-4 \\ 
\end{array} 
\]
Since the standard errors are comparably small, we shall ignore them in the sequel and compute only information values w.r.t. the means of the parameters. Furthermore, the risk premium $\l$ cannot significantly distinguished from $0$. Hence, we omit it. Recall, we are interested in the mutual informations 
\begin{align*}
&I(\sqrt{v_t \tau} : \vec r^t_{t - n\tau}) \\
&I(r_t : \vec r^0_{t - n\tau})
\end{align*}
between the volatility $\sqrt{v_t \tau}$ and previous returns
\[
\vec r^t_{t - n\tau} = r_{t-n\tau}, r_{t-(n-1)\tau}, \ldots, r_{t}\, ,
\]
and previous returns and their subsequent return, respectively. In section \ref{infoVol} we derived in proposition \ref{volRet} and \ref{PastFuture}
upper bounds $U_1$ for the mutual information $I(\sqrt{v_t \tau} : \vec r^t_{t - n\tau})$  and $U_2$ for the mutual information $I(r_t : \vec r^0_{t - n\tau})$, respectively. We computed in theorem \ref{oneTwo} and \ref{threeMut} these upper bounds in terms of the parameters of the SDE \eq{stochVol} and the stationary distribution of the volatility process. Since the parameters $\g$, $\t$ and $\k$ from \cite{sixFactor} are those for an annual time resolution, we have to choose
\begin{equation}
\tau = \dfrac{1}{252} \
\end{equation}
because we are dealing with daily returns. The previous corollaries provide explicit formulae for the upper bounds $U_1$ and $U_2$ which can be computed with ease in the case the stationary distribution exists. One checks that for the parameter values in the previous table, there is a stationary distribution in the cases $a = 0; b=3/2$, $a=1; b=1$ but not for the Heston model $a =0$ and $b = 1/2$. We list the values for the upper bounds $U_1$ and $U_2$ for the various models including the values $G_r$ and $G_{\sqrt{v}}$ in \eq{infoGap} at the end of section \ref{infoVol}.
\[
\begin{array}{ccccccccc}
a & b & \a & \b & U_2 & U_1 & G_r & G_{\sqrt{v}}\\ 
0 & 1/2 & 0.9598 & 18.35 &--  &  -- &  -- &  --\\ 
0 & 1 & 4.502 & 0.09526 &  0.1428 & 2.230 & 12.67 & 7.318\\ 
1 & 1 & 2.761 & 102.5 & 0.2839 & 2.506  & 16.14 & 7.552\\ 
0 & 3/2 & 8.511e-4 & -- &   0.1167 & 1.910& 11.53 & 7.343 \\ 
1 & 3/2 & 4.599 & 0.1008 & 0.1389 & 2.190 & 11.54 & 7.323
\end{array} 
\]
One sees, past returns never accumulate enough information about volatility. Dependencies between past and future returns are weak in all cases. Even though we have already observed in remark \ref{remthreeMut} that decreasing the time $\tau$ subsequent returns are read off does not affect the upper bound $U_2$ for the mutual information between past and future returns, the upper bound $U_1$ for the mutual information between past returns and the current volatility $v_t$ grows by increasing the observation frequency for the returns. Hence, instead of asserting that daily returns do not yield enough information about their volatility the subsequent table lists necessary returns per annum, i.e. $1/\tau$, to obtain enough information from returns about volatility.
\[
\begin{array}{ccccc}
a,b & 0,1 &1,1 & 0,3/2 & 1,3/2 \\ 
1/\tau& 6.62\, e6& 6.08 \, e6 & 1.32\, e7 & 7.24\, e6
\end{array} 
\]
In all cases we need return data which is quoted every two or three seconds.

\subsection{Exponential Ornstein-Uhlenbeck process}

\subsubsection{One-Factor Model}
Another popular mean-reverting model is \cite{Perello2008}
\begin{align*}
dS_t &=  S_t m e^{v_t} dW^0_t \\
dv_t &= - \g v_t dt + \k  dW^1_t \, .
\end{align*}
The process $v_t$ is an Ornstein-Uhlenbeck process. The solution of the corresponding Fokker-Planck equation with 
initial condition $v_0$ is 
\begin{equation} \label{eq:transition}
p(v_\tau|v_0) = \sqrt{\dfrac{2 \g}{2 \pi \k^2 \left( 1 - e^{-2\g\tau } \right)}}e^{- \dfrac{\g}{\k^2} \dfrac{\left(v_\tau - v_0 e^{-\g \tau} \right)^2}{1 - e^{- 2 \g \tau}}} \, ,
\end{equation}
that is, a normal distribution with mean $v_0 e^{-\g \tau} $ and variance 
\[
\dfrac{ \k^2}{2 \g} \left( 1 - e^{- 2 \g \tau} \right) \, .
\]
The Ornstein-Uhlenbeck process has the stationary distribution
\[
\rho(v) = \sqrt{\dfrac{\g}{\pi \k^2}} e^{- \g v^2/\k^2} \, ,
\]
that is, a normal distribution with mean $0$ and variance
\begin{equation} \label{eq:statVar}
\dfrac{\k^2}{2 \gamma} \, .
\end{equation}

\begin{cor} \label{ornstein}
We have
\begin{align*}
I(v_\tau : v_0)  &=  - \dfrac{1}{2} \log \left( 1 - e^{-2 \g \tau} \right)  \\
I(r_\tau : v_0) & \leq  \dfrac{2\g}{\k^2}
\end{align*}
\end{cor}

\begin{proof}
Since we know the analytical expression \eq{transition} for the transition probability, we can compute the mutual information $I(v_\tau: v_0)$ precisely without the proxy \ref{proxy} by means of example \ref{diffEntGauss}.
\begin{align*}
I(v_\tau : v_0) &= h(v_\tau) - h(v_\tau| v_0) \\
&= \dfrac{1}{2} \log \left( \pi e \dfrac{\k^2}{\g} \right) - \dfrac{1}{2} \left( \pi e \left( 1 - e^{-2 \g \tau} \right) \dfrac{\k^2}{\g}\right) \\
&= - \dfrac{1}{2} \log \left( 1 - e^{-2 \g \tau} \right)
\end{align*}
Furthermore, if we set $\rho(v) = e^{-V(v)}$ we obtain 
\[
V(v) = \dfrac{\g}{\k^2} v^2 + Z
\]
for some constant $Z$ and therefore
\[
V''(v) = \dfrac{2\g}{\k^2} \, .
\]
Hence, the stationary distribution fulfils the logarithmic Sobolev inequality with 
\[
\l = \dfrac{2\g}{\k^2} \, ,
\]
and we obtain from theorem \ref{threeMut}
\begin{align*}
I(r_\tau : v_0) &\leq  \dfrac{1}{2\l} \int 2 \left| \dfrac{d}{dv} \log m e^{v} \right|^2 \, d\rho(v) \\
&= \dfrac{\k^2}{2\g} \, .
\end{align*}
\end{proof}

From \cite{Perello2008} we obtain for 
\[
\tau = 1 \text{ day}
\]
the mean of the parameter values
\[
\g = 1.82 \, e-3 \, \text{days}^{-1}\qquad \k^2 = 1.4\, e-2 \, \text{days}^{-1} \qquad \rho = -0.4 \qquad m = 1.5\, e-3 \, \text{days}^{-1/2}
\] 
where $\rho dt = dW^0_t dW^1_t$.  Corollary \ref{ornstein} and theorems \ref{oneTwo} and \ref{threeMut}, respectively, yield
\[
I(me^{v_\tau} : \vec r^\tau_{\tau - n\tau}) \leq I(v_\tau : v_0) + I(r_\tau : v_\tau| v_0) =  - \dfrac{1}{2} \log \left( 1 - e^{-2 \g \tau} \right)   - \dfrac{1}{2} \log( 1 - \rho^2) = 2.9
\]
and 
\[
I(r_\tau : \vec r^\tau_{\tau - n\tau})  \leq I(r_\tau : v_0) \leq  \dfrac{2 \g}{\k^2} = 3.85 \, .
\]
The bound on the mutual information $ I(r_\tau : v_0) $ is very weak. Hence we compute the mutual information numerically. We obtain from example \ref{diffEntGauss}
\begin{align}
I(r_\tau : v_0) &= h(r_\tau) - h(r_\tau | v_0) \label{eq:numericalOne} \\
&= h(r_\tau) +  \int p(r_\tau | v_0) \log p(r_\tau | v_0) \,dr_\tau d\rho(v_0) \nn \\
&= h(r_\tau) - \dfrac{1}{2} \int \log \left( 2 \pi e  m^2  e^{2 v_0} \right) \, d\rho(v_0) \nn \\
&= h(r_\tau) - \dfrac{1}{2} \log ( 2 \pi e m^2) \nn \\
&= 0.86 \nn
\end{align}
where the entropy $h(r_\tau)$ was computed numerically, recalling that the distribution of $r_\tau$ is 
\begin{equation} \label{eq:numericalTwo}
p(r_\tau) = \int p(r_\tau| v_0) \rho(v_0) \, dv_0 = \int \dfrac{1}{\sqrt{2 \pi m^2 e^{2 v_0}}} e^{- \dfrac{r_\tau^2}{2 m^2 e^{2 v_0}}} \sqrt{\dfrac{\g}{\pi \k^2}} e^{- \g v_0^2/\k^2}  dv_0 \, .
\end{equation}
Since we deal with daily returns, we compute the necessary information $G_r$ and $G_{me^v}$ which we need for quoting the returns and the volatility up to $0.01$ percent precision as in \eq{infoGapDaily} and obtain
\[
G_r = \log \left( \dfrac{\sg_f}{\sg_M} \right) =  \log \left( \dfrac{m e^{\dfrac{\k^2}{4 \gamma}}}{\sg_M} \right) = 6.0
\]
and
\begin{align*}
G_{me^{v}} &= h(v) + \int \log (m e^{v}) \, d\rho(v) + \dfrac{1}{2}\log(252) -  \dfrac{1}{2} \log \left(  2 \pi e \sg_M^2 \right) \\
&= \dfrac{1}{2} \log \left(  2 \pi e \dfrac{\k^2}{2 \g}  \right) + \log ( m  ) + \dfrac{1}{2}\log(252) - \dfrac{1}{2} \log \left(  2 \pi e \sg_M^2 \right) \\
&= \log \left( \dfrac{\k}{\sqrt{2 \g} \sg_M} \right) + \log ( m ) + \dfrac{1}{2}\log(252) = 7.5 \, 
\end{align*}
with $\sg_M = \frac{1}{4}10^{-4}$. We summarize our results in the subsequent table. As previously,  $U_1$ denotes the upper bound for the mutual information $I(me^{v_\tau} : \vec r^\tau_{\tau - n\tau}) $ and $U_2$ the one for the mutual information $I(r_\tau : \vec r^\tau_{\tau - n\tau}) $. 
\[
\begin{array}{cccc}
 U_1 & U_2  & G_r & G_{me^v} \\ 
2.9 & 0.86 & 6.0 & 7.5 
\end{array} 
\]
As for the previous models, estimating volatility in the Exponential Ornstein-Uhlenbeck setting from stock data is hardly possible. Again returns are only weakly dependent in this model. To obtain the required precision for the volatility we would need
at least $9995$ returns a day, i.e., we need nearly secondly quoted returns.

\subsubsection{Two-Factor Model}
As we show now, the situation does not improve in case of multi-factor
models. In \cite{ALIZADEH2002} a two-factor version of the above
model, i.e.
\begin{align*}
dS_t &=  S_t m e^{v_{1,t} + v_{2,t}} dW^0_t \\
dv_{1,t} &= - \g_1 v_{1,t} dt + \k_1  dW^1_t \\
dv_{2,t} &= - \g_2 v_{2,t} dt + \k_2  dW^2_t
\end{align*}
is defined and compared to the above one-factor model. In this paper,
all Brownian motion are assumed to be independent which we adopt for
simplicity as well.

Regarding the information computations two complications arise when
considering multi-factor models:
\begin{itemize}
\item Approximating the information between $\vec v_t = (v_{1,t},
  v_{2,t})$ and $\vec r_{t-n \tau}^t$ by
  \[ I(\vec v_t : \vec r_{t-n \tau}^t) \leq I(\vec v_t : \vec
  v_{t-\tau}) + I(\vec v_t : r_t | \vec v_{t-\tau}) \] as in
  Prop.~\ref{volRet} would overstate the actual information since the
  returns do not depend on $v_{1,t}$ and $v_{2,t}$ directly, but only via
  the sum $v_{1,t} + v_{2,t}$ resulting in a loss of information.
\item While the two-dimensional process $\vec v_t$ is Markovian, this
  is no longer the case for the volatility process $\sg_t = m
  e^{v_{1,t} + v_{2,t}}$. Thus, we need to adapt the bound from
  Prop.~\ref{volRet} in order to obtain a tighter bound on $I(\sg_t :
  \vec r_{t-n \tau}^t)$ which is our prime interest.
\end{itemize}

Fortunately, the exponential Ornstein-Uhlenbeck model is fully
tractable and even though $w_t = v_{1,t} + v_{2,t}$ is not Markovian,
it is still a Gaussian process with stationary variance $\sg_w^2 =
\frac{\k_1^2}{2 \g_1} + \frac{\k_2}{2 \g_2}$ and covariance $c_{w,n} =
\frac{\k_1^2}{2 \g_1} e^{-\g_1 n \tau} + \frac{\k_2^2}{2 \g_2}
e^{-\g_2 n \tau}$ between $w_t$ and $w_{t + n \tau}$. In particular,
we can compute the conditional distribution of $w_t$ conditioned on
$\vec w_{t-n \tau}^{t-\tau}$ and thus the mutual information $I(w_t :
\vec w_{t-n \tau}^{t-\tau}) = I(\sg_t : \vec \sg_{t-n \tau}^{t-\tau})$.

\begin{prop}
\[
I(\sg_t : \vec r_{t-n \tau}^t) \leq  I(w_t : \vec w_{t-n \tau}^{t-\tau})
\]
\end{prop}

\begin{proof}
We can bound the information as follows:
\begin{eqnarray*}
  I(\sg_t : \vec r_{t-n \tau}^t) & = & I(w_t : \vec r_{t-n \tau}^t) \\
  & \leq & I(w_t : \vec w_{t-n \tau}^{t-\tau}, \vec r_{t-n \tau}^t) \\
  & = & I(w_t : \vec w_{t-n \tau}^{t-\tau}) + I(w_t : \vec r_{t-n \tau}^t | \vec w_{t-n \tau}^{t-\tau}) \\
  & = & I(w_t : \vec w_{t-n \tau}^{t-\tau}) + I(w_t : \vec r_{t-n \tau}^{t-\tau} | \vec w_{t-n \tau}^{t-\tau})
  + I(w_t : r_t | \vec w_{t-n \tau}^{t-\tau}, \vec r_{t-n \tau}^{t-\tau}) \\
  & = & I(w_t : \vec w_{t-n \tau}^{t-\tau}) + I(w_t : r_t | \vec w_{t-n \tau}^{t-\tau})
\end{eqnarray*}
where we have used that $w_t \independent \vec r_{t-n \tau}^{t-\tau} |
\vec w_{t-n \tau}^{t-\tau}$, i.e. the past returns have no influence
on $w_t$ when the whole history $\vec w_{t-n \tau}^{t-\tau}$ is
available. Note that as we assumed no leverage effect, i.e. $\rho_1 = \rho_2 =
0$, the last term vanishes. Thus, we are left with the bound based on
the information structure of the Gaussian process.
\end{proof}

Table \ref{tab:TwoFactor} contains numeric values for the information $I(w_t
: \vec w_{t-n \tau}^{t-\tau})$ using parameters from
\cite{ALIZADEH2002}:
\[ 
m = 2.32 \, e-3 \text{ days}^{-\frac{1}{2}} \quad \g_1 = 2.02 \,
e-2 \text{ days}^{-1} \quad \k_1^2 = 4.13 \, e-3 \text{ days}^{-1}
\]
\[
\g_2 = 1.43 \text{ days}^{-1} \quad \k_2^2 = 4.14 \, e-2\text{ days}^{-1} 
\] 
where all parameters have been converted to daily units based on their convention 
of $257$ trading days per year.  For illustration we have chosen the parameter values 
fitted on Canadian dollar exchange rates, since these give rise to the largest information 
values.

\begin{table}[h]
  \centering
  \begin{tabular}{c|ccccccc}
    & \multicolumn{7}{c}{history length $n$} \\
    & 1 & 2 & 3 & 4 & 5 & 10 & 100 \\ \hline
    $I(\sg_t : \vec r_{t-n \tau}^t)$ & 0.778 & 0.819 & 0.835 & 0.842 & 0.845 & 0.847 & 0.847
  \end{tabular}
  \caption{Upper bound on the mutual information (in nats) about the volatility when observing returns $\vec r_{t-n \tau}^t$.}
  \label{tab:TwoFactor}
\end{table}

Finally, we compute the mutual information $I(r_\tau: w_0)$ with $w_0 = v_{1,0} + v_{2,0}$ in the same manner as we computed the mutual information $I(r_\tau: v_0)$ for the single-factor mode in \eq{numericalOne}, namely numerically. We proceed by replacing $v_0$ by $w_0$ and the variance  $\frac{\k^2}{2\g}$ by $\frac{\k_1^2}{2 \g_1} + \frac{\k_2^2}{2 \g_2}$ in formula \eq{numericalTwo} and obtain the value
\[
I(r_\tau: w_0) =  0.093 \, .
\]
As before, we compute the required information $G_r$ and $G_{me^{w}}$
as in \eq{infoGapDaily}, again using parameters with daily units and obtain
\[
G_r = \log \left( \dfrac{\sg_f}{\sg_M} \right) = \log \left(\dfrac{me^{\dfrac{\k_1^2}{4 \g_1} + \dfrac{\k_2^2}{4 \g_2}}}{\sg_M} \right) = 4.6 
\]
and
\begin{align*}
G_{me^{w}} &= h(w) + \int \log (m e^{w}) \, d\rho(w) + \dfrac{1}{2}\log(252)  - \dfrac{1}{2} \log \left(  2 \pi e \sg_M^2 \right) \\
&= \frac{1}{2} \log \left( \dfrac{\frac{\k_1^2}{2 \g_1} + \frac{\k_2}{2 \g_2}}{\sg_M^2} \right) + \log ( m ) + \dfrac{1}{2}\log(252) = 6.2
\end{align*}
where $\sg_M = \frac{1}{4}10^{-4}$.

Here, in contrast to the one-factor model the bound on the information
increases when longer histories are observed. Nevertheless, the
information values are much smaller, also in relation to the required
information, than for the one-factor model considered above and the
information essentially saturates after about 10 days. As explained in
\cite{ALIZADEH2002} the two-factor model improves the one-factor
model by utilizing two very different time scales for the processes
$v_{1,t}$ and $v_{2,t}$. In particular, one of these processes
captures fast and transient changes of the volatility while the other
models long-range dependencies. Especially the transient process
diminishes the temporal dependence of the volatility process
substantially. Thus, if the finding on the nature of two-factor models
holds up in general, i.e. across asset classes, we would not be
surprised if they provide even less information about the hidden
volatility process than their, potentially misspecified, one-factor
relatives.

\section{Conclusions}

We developed a general information theoretical frame to estimate in 
stochastic volatility models the uncertainty about the hidden
volatility when it is inferred from stock data. This frame also allows
to quantify the dependencies between subsequent returns
in these models.\\
In single factor models, quoting volatility up to $0.01$ percent needs in
general at least secondly quoted return time series. Even then,
only the upper bound, which we derived for the mutual information
between past returns and present volatility, bridges the information gap
computed in \eq{infoGap} not necessarily the actual mutual information itself. The 
situation does not improve when considering two-factor models. To the contrary, as we show
in the case of a two-factor exponential Ornstein-Uhlenbeck model, even less information about the 
underlying volatility can be recovered from return data. The reason being that real volatility is
a highly varying process, which is better captured by the more flexible two-factor model. Thus, 
instantaneous volatility cannot be estimated from much more than about ten days of return data.
We note that these results also apply when predicting volatility, i.e. after about ten days the best prediction
is based on the stationary distribution alone. This high intrinsic uncertainty of volatility estimates
sheds doubt on the standard practice of comparing volatility models based on their forecast performance.
To the least, much care is needed to obtain reliable and significant statements about the relative performance
of different models. To our knowledge, this point has not been discussed in the literature let alone
being studied in a rigorous and quantitative fashion. Here, we demonstrate that
our information theoretic frame is ideally suited to address this issue and derive
precise data requirements, e.g. returns at second resolution, in order to obtain reliable
volatility estimates. \\
Remarkable, from a technical perspective, is the upper bound we derived in theorem \ref{threeMut} for 
the mutual information between past returns
and future ones by means of the logarithmic Sobolev inequality. Deriving
such an inequality is interesting in its own right because it is in general 
pretty difficult to derive upper bounds for the mutual information because
it captures all dependencies between two random variables and is therefore
much harder to handle than, for instance, correlation. The use
of the logarithmic Sobolev inequality demonstrates a nice interplay between
recent developments in statistical physics where this inequality played a 
prominent role in estimating the convergence rate of perturbed thermodynamical
systems towards their equilibrium and financial mathematics. As far as 
we know, the present paper is the first one which applies the logarithmic
Sobolev inequality in finance. Apart from its mathematical charm the
inequality yields an upper bound for the mutual information between 
past and future returns which does no longer depend on the time
resolution $\tau$ and which is quite tight for the generic stochastic volatility models considered in 
section \ref{sixModels}. Thus proving that despite the correlated volatility process returns are only 
weakly dependent in stochastic volatility models. This is not only consistent with the observation that
returns are hard to predict, but 
might also motivate the
use of simpler, single factor jump-diffusion models, which are analytically more tractable.\\
We adopt in a forthcoming paper the information theoretical methods developed
in the present one for estimating the information content of vanilla option prices
about the underlying volatility. It turns out that volatility estimates from option 
prices are far superior over stock returns as more data is available: for the same
volatility different options at different strikes and maturities are quoted, and 
furthermore the relationship between option prices and volatility is much more direct than between returns and volatility.
Future research will be also devoted on the question to which extend \textit{rough paths}
stochastic volatility models, see for instance \cite{roughPath}, which have recently
caught attention, are in the information theoretical scope of the present paper.

\section{Funding}

Nils Bertschinger and Oliver Pfante thank Dr. h.c. Maucher for funding their positions.

\bibliographystyle{abbrv}
\bibliography{Heston}

\begin{thebibliography}{10}

\bibitem{ALIZADEH2002}
S.~Alizadeh, M.~W. Brandt, and F.~X. Diebold.
\newblock Range-based estimation of stochastic volatility models.
\newblock {\em Journal of Finance}, 57:1047--1091, 2002.

\bibitem{Bakry}
D.~Bakry and M.~Ledoux.
\newblock L\'{e}vy-gromov's isoperimetric inequality for an
  infinite-dimensional diffusion generator.
\newblock {\em Invent. Math.}, 123(2):259--281, 1996.

\bibitem{roughPath}
C.~Bayer, P.~K. Friz, and J.~Gatheral.
\newblock Pricing under rough volatility.
\newblock {\em Quantitative Finance}, 16(6):887--904, 2016.

\bibitem{Black1973}
F.~Black and M.~Scholes.
\newblock {The Pricing of Options and Corporate Liabilities}.
\newblock {\em Journal of Political Economy}, 81(3):637, 1973.

\bibitem{Bollerslev2006}
T.~Bollerslev, J.~Litvinova, and G.~Tauchen.
\newblock {Leverage and volatility feedback effects in high-frequency data}.
\newblock {\em Journal of Financial Econometrics}, 4(3):353--384, 2006.

\bibitem{Bouchaud2001}
J.~P. Bouchaud, a.~Matacz, and M.~Potters.
\newblock {Leverage effect in financial markets: the retarded volatility
  model.}
\newblock {\em Physical review letters}, 87(22):228701, 2001.

\bibitem{Bouchaud2003}
J.-P. Bouchaud and M.~Potters.
\newblock {Theory of Financial Risk and Derivative Pricing}.
\newblock {\em Theory of Financial Risk and Derivative Pricing}, page 379,
  2003.

\bibitem{sixFactor}
P.~Christoffersen, K.~Jacobs, and K.~Mimouni.
\newblock Models for s{\&}p 500 dynamics: Evidence from realized volatility,
  daily returns, and option prices.
\newblock {\em Review of Financial Studies}, 23(8), 2007.

\bibitem{Cover2006}
T.~M. Cover and J.~A. Thomas.
\newblock {\em Elements of Information Theory}.
\newblock Wiley-Interscience, 2 edition, July 2006.

\bibitem{Ding1993}
Z.~Ding, C.~W. Granger, and R.~F. Engle.
\newblock {A long memory property of stock market returns and a new model}.
\newblock {\em Journal of Empirical Finance}, 1(1):83--106, 1993.

\bibitem{Feller1951}
W.~Feller.
\newblock {Two Singular Diffusion Problems}.
\newblock {\em Annals of Mathematics}, 54(2):286--295, 1951.

\bibitem{Gatheral2010}
J.~Gatheral.
\newblock Jump-diffusion models.
\newblock {\em Encyclpedia of Quantitative Finance}, 2010.

\bibitem{Gross}
Gross.
\newblock Logarithmic sobolev inequalities.
\newblock {\em Amer. J. Math.}, 97:1061--1083, 1975.

\bibitem{Hansen2005}
P.~R. Hansen and A.~Lunde.
\newblock A forecast comparison of volatility models: does anything beat a
  garch(1,1)?
\newblock {\em Journal of Applied Econometrics}, 20(7):873--889, 2005.

\bibitem{Jones}
C.~S. Jones.
\newblock The dynamics of stochastic volatility: evidence from underlying and
  options markets.
\newblock {\em Journal of Econometrics}, (116):181--224, 2003.

\bibitem{Kraskov2004}
A.~Kraskov, H.~St{\"o}gbauer, and P.~Grassberger.
\newblock Estimating mutual information.
\newblock {\em Phys. Rev. E}, (69), 2004.

\bibitem{LeBaron2001}
B.~LeBaron.
\newblock {Stochastic Voltility as a Simple Generator of apparent financial
  power laws and long memory}.
\newblock {\em Quantitative Finance}, 6(1):627--631, 2001.

\bibitem{Lo1991}
A.~W. Lo.
\newblock {Long Term Memory in Stock Market Prices}.
\newblock {\em Econometrica}, 59(5):1279--1313, 1991.

\bibitem{Villani}
P.~A. Markowich and C.~Villani.
\newblock On the trend to equilibrium for the fokker-planck equation : An
  interplay between physics and functional analysis.

\bibitem{Muzy2000}
J.~F. Muzy, J.~Delour, and E.~Bacry.
\newblock {Modelling fluctuations of financial time series : from cascade
  process to stochastic volatility model}.
\newblock {\em Eur. Phys. J. B}, 17:537--548, 2000.

\bibitem{Perello2008}
J.~Perello, R.~Sircar, and J.~Masoliver.
\newblock {Option pricing under stochastic volatility: the exponential
  Ornstein-Uhlenbeck model}.
\newblock page~26, 2008.

\bibitem{Risken}
H.~Risken.
\newblock {\em The Fokker-Planck Equation}.
\newblock Springer, 1996.

\bibitem{Schoutens2003}
W.~Schoutens.
\newblock {\em L\'{e}vy Processes in Finance: Pricing Financial Derivatives}.
\newblock Wiley, 2003.

\bibitem{Shannon1948}
C.~Shannon.
\newblock A mathematical theory of communication.
\newblock {\em The Bell System Technical Journal}, 27:379--423, 1948.

\bibitem{Stam}
A.~J. Stam.
\newblock Some inequalities satisfied by the quantities of information of
  fisher and shannon.
\newblock {\em Information and Control}, 2:101--112, 1959.

\end{thebibliography}

\end{document}